\documentclass[12pt]{article}
\usepackage[english]{babel}
\usepackage[utf8]{inputenc}
\usepackage{amsmath,amsfonts,amssymb,amsthm}
\usepackage{enumerate}
\usepackage{setspace}
\usepackage[plain]{fullpage}
\usepackage[authoryear, longnamesfirst]{natbib}
\usepackage{mathrsfs}  
\usepackage{calligra}
\usepackage{hyperref}
\usepackage{adjustbox}
\usepackage{subcaption}
\usepackage{amsmath}
\usepackage{color}
\usepackage{graphicx}
\usepackage{algorithm}
\usepackage{algorithmic}
\usepackage[longnamesfirst]{natbib}
\usepackage{dsfont}

\theoremstyle{definition}

\newtheorem{definition}{Definition}

\theoremstyle{plain}
\newtheorem{prop}{Proposition}
\newtheorem{thm}{Theorem}
\newtheorem{lemma}{Lemma}

\theoremstyle{remark}
\newtheorem{remark}{Remark}

\DeclareMathOperator{\Yc}{\mathcal{Y}}
\DeclareMathOperator{\R}{\mathbb{R}}
\DeclareMathOperator{\E}{\mathbb{E}}
\DeclareMathOperator{\Xc}{\mathcal{X}}
\DeclareMathOperator{\power}{\text{power}}

\title{Approximate Least-Favorable Distributions and Nearly Optimal Tests via Stochastic Mirror Descent\footnote{We would like to thank Isaiah Andrews, Tim Armstrong, Patrik Guggenberger, Phillip Heiler, Ulrich M\"uller, Mikkel S\o lvsten, Karthik Sridharan, Vasilis Syrgkanis and participants of Aarhus Workshop in Econometrics VII for very helpful feedback, comments,
and suggestions. We gratefully acknowledge financial support from the NSF under grants SES-2315600, 2229012, 2312204, 2403007; and from the Department of Defense through the Air Force Office of Scientific Research under award number FA9550-20-1-0397 and ONR 1398311.}}

\author{
Andr\'es Aradillas Fern\'andez\thanks{Department of Economics, Massachusetts Institute of Technology.} \and \and Jos\'e Blanchet\thanks{Management Science and Engineering Department, Stanford University.} \and Jos\'e Luis Montiel Olea\footnotemark[4]   \and Chen Qiu\thanks{Department of Economics, Cornell University.} \and J\"org Stoye\footnotemark[4] \and Lezhi Tan\footnotemark[3]  
}

\date{November 2025}

\begin{document}

\maketitle

\onehalfspacing

\begin{abstract}  
We consider a class of hypothesis testing problems where the null hypothesis postulates $M$ distributions for the observed data, and there is only one possible distribution under the alternative. We show that one can use a \emph{stochastic mirror descent routine} for convex optimization to \emph{provably} obtain---after finitely many iterations---both an approximate least-favorable distribution and a nearly optimal test, in a sense we make precise. Our theoretical results yield concrete recommendations about the algorithm's implementation, including its initial condition, its step size, and the number of iterations. Importantly, our suggested algorithm can be viewed as a slight variation of the algorithm suggested by \cite{elliott2015nearly}, whose theoretical performance guarantees are unknown. 
\end{abstract}

\doublespacing

\section{Introduction} \label{sec:introduction} 

Consider the problem of testing a null hypothesis that postulates \emph{finitely many} distributions for the observed data against a \emph{single} alternative hypothesis. More concretely, suppose the data is modeled as a $\mathcal{Y}$-valued random variable, denoted as $Y$, and let the hypothesis testing problem take the form:
\begin{equation} \label{eqn:testing_problem_intro}
\mathbf{H}_0: \textrm{the distribution of $Y$ is $F_{m}$, \: $m =1, \ldots, M,$} \quad \textrm{ \emph{vs.} } \quad \mathbf{H}_1: \textrm{the distribution of $Y$ is $G$.} 
\end{equation}
The structure of the null and alternative hypothesis in \eqref{eqn:testing_problem_intro} arises in nonstandard testing problems that involve nuisance parameters; we refer the reader to the work of \cite{elliott2015nearly}.\footnote{See also \cite{elliott2014pre}, \cite{muller2016measuring}, \cite{muller2017fixed}, \cite{muller2018long}, \cite{guggenberger2019more}, \cite{muller2020low}, \cite{dou2021generalized}, \cite{li2021linear}, \cite{muller2025more}, \cite{muralidharan2025factorial} for related problems in different applications.} 

We are interested in the computational aspects of finding \emph{the most powerful test of size $\alpha$} for the testing problem in  \eqref{eqn:testing_problem_intro}.\footnote{The most powerful test of size $\alpha$ correctly rejects the null hypothesis as frequently as possible, while guaranteeing that the probability of incorrectly rejecting the null is bounded above by a prespecified constant $\alpha \in (0,1)$.} It is well known that such most powerful test is based on the likelihood ratio that replaces the null hypothesis by a single mixture distribution obtained from averaging $F_1, \ldots, F_M$ according to what the literature refers to as the \emph{least-favorable distribution}; see, for example, Theorem 1 in \cite{lehmann1948most}. 
%and also \cite{lehmann2005testing}, Chapter 3, p. 84
It is also known, but perhaps less so, that the least favorable distribution can be characterized as the solution of a nonlinear convex program; see Lemma 3 in \cite{krafft1967optimale}. Although these results immediately suggest a computational strategy for finding the most powerful test of size $\alpha$ for \eqref{eqn:testing_problem_intro}, the algorithms that have been thus far suggested in the econometrics literature for this purpose---for example, \cite{chiburis2008approximately}, \cite{moreira2013contributions}, \cite{elliott2015nearly}---do not explicitly exploit the connection between finding the least-favorable distribution and solving a nonlinear convex program.

Our first result (Theorem \ref{thm:Mullers_algo_is_SMD}) shows that the structure of the convex program that defines the least-favorable distribution associated with the testing problem in \eqref{eqn:testing_problem_intro} is amenable to the application of the methods of \emph{(stochastic) mirror descent} (henceforth, S-MD); see Chapter 6.1 in \cite {bubeck2015convex} and also Chapter 3 in \cite{nemirovski1983problem} for textbook references. The methods of mirror descent are a family of iterative procedures designed for finding an  \emph{approximate solution} of convex problems in high dimensions. These methods require repeated (but finitely many) evaluations of an unbiased estimator of the subgradient of the objective function. These methods exploit the geometry of the optimization domain, but they do not exploit any higher-order smoothness information about the optimization problem.\footnote{Mirror descent is known to outperform gradient descent in some optimization domains, such as the probability simplex; see Section 4.3 of \cite{bubeck2015convex}. See also Section 5.1.1 in \cite{nemirovski2009robust}.}  We have emphasized the term \emph{approximate solution} since a common approach in theoretical computer science, operations research, and optimization is  \emph{``to relax the requirement of finding an optimal solution, and instead settle for a solution that is good enough}''  \cite[p. 14]{shmoys2011design}. We follow this approach and, instead of insisting on finding the exact solution of the convex program that defines the least favorable distribution, we settle for an \emph{approximate least-favorable distribution} that solves the convex program of interest, but up to a prespecified additive factor $\epsilon$ (see our Definition \ref{def:epsilon-least-favorable}). 

Our second result (Theorem \ref{thm:epsilon_least_favorable}) shows that our suggested algorithm can indeed be used to \emph{provably} find---after finitely many iterations and with high probability---an approximate least-favorable distribution. Our results are probabilistic since the output of the algorithm depends on the stochastic estimator of the subgradient, and not the  subgradient itself. Our theoretical results yield concrete recommendations about the algorithm's implementation, including its initial condition, its step size, the number of iterations, and the number of stochastic draws per iteration that can be used to estimate the subgradient of the objective function. We think there at least two important implications from our theoretical results. First, the number of iterations used by the algorithm---$ \left( 4(1 - \alpha)^2 / \alpha^2 \epsilon^2 \right) \cdot \ln(M) $---scales logarithmically in $M$, which means there is no theoretical sense in which the number of iterations should be expected to scale poorly as function of how many elements there are in the null hypothesis in \eqref{eqn:testing_problem_intro}. Second, and perhaps the most striking feature of our analysis, the unbiased estimator of the gradient used by our S-MD routine---defined in Equation \eqref{eq:gradient_Theorem1} in Theorem \ref{thm:Mullers_algo_is_SMD}---\emph{can be based on a single Monte Carlo draw} from each null distributions $F_{m}$, $m=1, \ldots, M$. As expected, taking a larger number of draws improves the approximation error of the S-MD routine, but using a small number of draws reduces the computational burden of the algorithm. 

Interestingly, our suggested S-MD routine can be viewed as a slight modification of the algorithm suggested by  \cite{elliott2015nearly} for finding a least-favorable distribution for the testing problem in \eqref{eqn:testing_problem_intro}.  We think this is an important observation since, to the best of our knowledge, the theoretical guarantees for their procedure remain unknown. We note, however, that---in addition to the recommended initial condition, the step size, the number of iterations, and the number of stochastic draws per iteration---there is at least one important conceptual difference between the two algorithms. As we will explain later, the output of our recommended S-MD algorithm is based on \emph{averaging} the output of the mirror descent routine over all iterations, as opposed to just trying to obtain the least-favorable distribution from the last update as suggested by \cite{elliott2015nearly}.\footnote{Averaging the trajectories of a stochastic gradient-descent routine is commonly referred to as Polyak-Ruppert averaging. See \cite{ruppert1988efficient} and \cite{polyak1992acceleration}. See also \cite{forneron2024estimation} for a discussion of Polyak-Ruppert averaging in the context of estimation and inference by stochastic optimization of nonlinear econometric models.} While it might be possible to derive theoretical guarantees for the last iterate, this will likely require considering step sizes that vary with each iteration (and in our case, the step size is fixed throughout the whole routine).

Finally, we analyze the extent to which the output of the S-MD routine allows us to construct the most powerful test for the problem in  \eqref{eqn:testing_problem_intro}. If it were possible to extract an exact least-favorable distribution, the most powerful test could be obtained directly. However, since we only obtain an approximate least-favorable distribution we need additional work to show that we can construct a \emph{nearly optimal} test. We define a test $\varphi$ to be $(\epsilon,\delta)$-nearly optimal (Definition \ref{def:nearly_optimal}) whenever i) the size of $\varphi$ is at most $\alpha +\delta$; and ii) up to the additive constant $\epsilon$, the test $\varphi$ rejects the null hypothesis as frequently as the most powerful test of size $\alpha$. Based on this definition, we say that a test is \emph{nearly optimal} if it is $(\epsilon,\delta)$-nearly optimal for some parameters $\epsilon$ and $\delta$. Our third result (Theorem \ref{thm:average_test}) shows that it is indeed possible to use the S-MD routine to construct a \emph{randomized} nearly optimal test with high probability. The nearly optimal test is shown to be the \emph{average} of each of the tests of the Neyman-Pearson form associated with each vector of multipliers generated by the S-MD routine. We present explicit expressions for its size distortion and its power loss relative to the best test. We note that an important challenge in reporting such a test is that, in principle, it requires keeping track of the history of multipliers obtained from the S-MD routine. When both $M$ and the number of iterations of S-MD are large, this could come at a significant computational cost. To address this issue, we show that there is a simple strategy to implement the average test: we randomize the number of iterations uniformly between $1$ and our recommended $T$, update the multipliers, and use the resulting test to decided whether or not to reject the null hypothesis. The resulting procedure, thus, can be thought of as the Neyman-Pearson (or likelihood ratio test) associated with the last iterate of the S-MD routine that stops randomly before our suggested number of rounds.       

{\scshape Related Literature:} When the sample space $\mathcal{Y}$ is infinite,  the mathematical problem that defines the most powerful test of a given size in the problem \eqref{eqn:testing_problem_intro} is an \emph{infinite} linear programming problem:  Since the testing problems analyzed in this paper posit $M$ null distributions for the observed data, the corresponding linear program has finitely many constraints but a choice variable of infinite dimension; see Section \ref{sec:notation} below for details. If the data were discrete-valued or if we were to discretize it---as suggested by \cite{chiburis2008approximately}, \cite{moreira2013contributions}, and \cite{moreira2019optimal}---then one could find the most powerful test of a given size by using any algorithm for finite linear programming. \cite{krafft1967optimale} is the seminal reference for using linear programming methods to characterize the most powerful test of a given size.

\cite{elliott2015nearly} also show how the most powerful test for composite hypotheses can be expressed as a minimax decision problem where a false rejection of $\mathbf{H}_0$ induces a loss of 1, and a false rejection of $\mathbf{H}_1$ induces a loss of $\phi>0$ (correct choices have loss of zero). In this problem the decision maker chooses a test $\varphi$. An adversarial nature decides which element in $\{F_1, F_2, \ldots, G\}$ to use to generate the data; consequently, a mixed strategy for nature can be represented by a vector in the simplex of $\mathbb{R}^{M+1}$. One could then use an algorithm for solving the corresponding maximin problem---for example, the Hedge algorithm suggested by \cite{fernandez2024epsilon}; the fictitious play algorithm suggested by \cite{guggenberger2025numerical}; or a general convex optimization routine as in \cite{chamberlain2000econometric}. An important limitation of this approach is that one would need to solve the maximin problem repeatedly for different values of $\phi$, until one finds a test with correct size. 

Although there are no theoretical results showing that one of these algorithms is better than another for the purpose of finding a nearly optimal test, we think that our analysis illustrates how a rich literature in optimization can be leveraged to provide theoretical results about the performance of different algorithms and also to provide practical recommendations regarding their implementation.

{\scshape Outline:} The rest of the paper is organized as follows. Section \ref{sec:notation} presents notation, the statement of the hypothesis testing problem of interest, and the primal and dual optimization problems that arise when searching for the most powerful test of a given size. Section \ref{sec:main_results} presents a formal definition of a stochastic mirror descent routine (S-MD) for convex optimization and also our main results: namely, that the algorithm provably generates an approximate least favorable distribution and a nearly optimal test. Section \ref{section:illustrative_example} uses an elementary testing problem that arises in the context of the univariate Gaussian location model to illustrate our main results. Section \ref{sec:extensions} discusses some extensions and Section \ref{sec:conclusion} concludes. The proofs of our main theorems and supporting lemmas are collected in Appendix \ref{online_appendix}. Additional results are collected in Appendix \ref{section:appendix_additional_results}. 

\section{Notation and Statement of the Problem} \label{sec:notation} 
We first present the formal statement of the hypothesis testing problem analyzed in this paper. We follow the notation and terminology used in \cite{elliott2015nearly} as close as possible. We then present the convex program associated with the least-favorable distribution.

\subsection{Statement of the Hypothesis Testing Problem} 

We observe a random element $Y$ that takes values in some space $\mathcal{Y}$ endowed with $\sigma$-algebra $\mathcal{F}$. Let $\nu$ denote a $\sigma$-finite measure defined over the measurable space $(\mathcal{Y},\mathcal{F})$. 
Let $F_1, \ldots, F_{M}$ denote $M>1$ candidate probability measures for the distribution of $Y$ under the null hypothesis. Let $G$ be the candidate distribution of $Y$ under the alternative hypothesis. Assuming all of these distributions are absolutely continuous with respect to $\nu$, Theorem 5.5.4 in \cite{dudley02} guarantees the existence of nonnegative integrable functions $f_1, \ldots, f_M, g$, which can be taken as the probability density functions of $F_1, \ldots, F_{M}, G$ relative to $\nu$. 

Based on a single observation of $Y$, the testing problem of interest is 
\begin{equation} \label{eqn:testing_problem_section2}
\mathbf{H}_0: \textrm{ the density of $Y$ is $f_{m}$, \: $m =1, \ldots, M,$} \quad \textrm{ \emph{against} } \quad \mathbf{H}_1: \textrm{ the density of $Y$ is $g$.} 
\end{equation}
Using the typical jargon of hypothesis testing problems, the null hypothesis in \eqref{eqn:testing_problem_section2} is \emph{composite}, since it contains more than one possible distributions for the data. The alternative hypothesis is \emph{simple}, in that it contains a single distribution.\footnote{As explained in Section 2.2 of \cite{elliott2015nearly}, the density $g$ can arise by appealing to the weighted average power criterion in cases where the alternative hypothesis is composite as well.} 

A \emph{statistical test} for \eqref{eqn:testing_problem_section2} (or simply a \emph{test}) is a measurable function $\varphi:\mathcal{Y} \rightarrow [0,1]$, where $\varphi(y)$ is interpreted as the probability of rejecting the null hypothesis given that data $y$ were observed. A test $\varphi$ is said to be \emph{nonrandomized} if $\varphi(y) \in \{0,1\}$ for $\nu$-almost every realization of $Y$; otherwise the test is said to be \emph{randomized}. 

The rate of Type I error under $f_{m}$ is the probability of rejecting the null hypothesis when $Y \sim f_{m}$ and it equals $\int \varphi f_{m} d \nu$. As usual, the \emph{size} of a test is the largest rate of Type I error under the null hypothesis. The \emph{power} of a test is the probability of rejecting the null hypothesis when $Y \sim g$ and it equals $\int \varphi g d \nu$. 

\subsection{Primal and Dual Problems in Hypothesis Testing} 

We would like to find the \emph{most powerful test} of size $\alpha$ for the problem \eqref{eqn:testing_problem_section2}. By definition, such a test correctly rejects the null hypothesis as frequently as possible, but guarantees that the probability of incorrectly rejecting the null is bounded above by the prespecified constant $\alpha \in (0,1)$. Mathematically, the problem of finding the most powerful test of size $\alpha$ can be written as:
\begin{equation}\label{eq:main_problem}
    \sup_{\varphi:\mathcal{Y} \rightarrow [0,1]} \int \varphi g d\nu, \quad \text{s.t.} \quad \int \varphi f_{m}d\nu \leq \alpha, \hspace{1em} m = 1,...,M.
\end{equation}
We refer to the optimization problem in \eqref{eq:main_problem} as the \emph{primal} problem associated with the hypothesis testing problem in \eqref{eqn:testing_problem_section2}. We note that the primal problem is an infinite linear programming problem, in the sense of \cite{anderson1987linear}. The infinite linear program in \eqref{eq:main_problem} has finitely many constraints but a choice variable of infinite dimension.

Define the Lagrangian function associated with the optimization problem \eqref{eq:main_problem} as 
\begin{equation}
    \label{eq:lagr_def}
        L(\varphi, \kappa) \equiv \int \varphi g d\nu -  \sum_{m=1}^M \kappa_m \left[\int \varphi f_{m}d\nu - \alpha\right],
\end{equation}
where we refer to $\kappa \equiv (\kappa_1,...,\kappa_M) \in \mathbb{R}_{+}^{M}$ as the Lagrange multipliers (or simply, \emph{multipliers}) associated with each of the inequality constraints in the primal problem \eqref{eq:main_problem}. 

Consider thus the optimization problem on $\mathbb{R}_{+}^{M}$ with variable $\kappa \equiv (\kappa_1,...,\kappa_M)$ given by
\begin{equation}
\label{eq:dual}
     \bar{v} \equiv \inf_{\kappa \in \mathbb{R}_{+}^{M}} f(\kappa),
\end{equation}
where
\begin{equation}\label{eqn:NP_oracle}     
     f(\kappa) \equiv \sup_{\varphi:\mathcal{Y} \rightarrow [0,1]} L(\varphi,\kappa). 
\end{equation}
We refer to the problem in \eqref{eq:dual} as the \emph{dual} problem of \eqref{eq:main_problem}. 

\begin{remark} \label{remark:dual_primal_relation} The dual problem in \eqref{eq:dual} can be viewed as a device to solve the primal problem in \eqref{eq:main_problem}. This is a well-known fact---see Lemma 3  in \cite{krafft1967optimale} and also \cite{cvitanic2001generalized,rudloff2010testing}---and we present a heuristic argument to help the exposition (relegating technical details to Appendix \ref{subsection:appendix_duality}). Suppose that the multipliers $\kappa^*$ solve the problem \eqref{eq:dual} in that $f(\kappa^*) = \bar{v}$. Then, by definition 
\begin{eqnarray*}
f(\kappa^*) &=& \sup_{\varphi:\mathcal{Y} \rightarrow [0,1] } L(\varphi, \kappa^*) \\ 
&= &  \int \varphi_{\kappa^*} g d\nu -  \sum_{m=1}^M \kappa^*_m \left[\int \varphi_{\kappa^*} f_{m}d\nu - \alpha\right], 
\end{eqnarray*}
where $\varphi_{\kappa^*}$ is a test of the \emph{Neyman-Pearson} form; that is 
\begin{equation} \label{eqn:test_NP_form} 
\varphi_{\kappa^*}(y) \equiv 
\left\{
\begin{array}{cc}
1 & \text{if } g(y) > \sum_{m=1}^{M} \kappa^*_m f_m(y)  \\
0 & \text{if } g(y) \leq \sum_{m=1}^{M} \kappa^*_m f_m(y). 
\end{array}
\right.
\end{equation}
Lemma 1 in \cite{elliott2015nearly} and Theorem 3.8.1 in \cite{lehmann2005testing} imply that if the test $\varphi_{\kappa^*}$ has size $\alpha$ under $\mathbf{H}_0$
then $\varphi_{\kappa^*}$ solves \eqref{eq:main_problem}; that is, 
it maximizes power among all tests of size at most $\alpha$. The direction of the vector $\kappa^*$---namely, $\lambda^* \equiv \kappa^* / \sum_{m=1}^{M} \kappa^*_m$ is a least-favorable distribution in the sense of \cite{lehmann2005testing}, Chapter 3, p. 84. \hfill \qed 
\end{remark}

\begin{remark}
In order to justify the terminology of primal and dual problems, Section \ref{subsection:appendix_duality} in Appendix \ref{section:appendix_additional_results} formalizes the connection between the optimization problems \eqref{eq:main_problem} and \eqref{eq:dual}, by showing that the value functions of both problems are equal, and that a solution to the dual problem in \eqref{eq:dual} can indeed be translated to a solution to the primal problem in  \eqref{eq:main_problem}. As usual, an important step in showing that the solution of the dual problem can be used to solve the primal problem consists in verifying that the complementary slackness conditions in the dual problem are satisfied. We also note that similar duality results have been established and used elsewhere; for example, see Lemma 3 in \cite{krafft1967optimale} and Proposition 3.1 and Equation 3.11 of \cite{cvitanic2001generalized}. Since these results are established with more generality than what is required by our framework, Appendix \ref{section:appendix_additional_results} presents simpler versions of these results.  
\end{remark}

\subsection{Solving the dual problem} 
We have explained how the solution to the dual problem in \eqref{eq:dual} can be used to construct the most powerful test of a given size. We now discuss the computational aspects of solving the dual problem. We start by showing that the objective function in \eqref{eq:dual} is convex over  $\mathbb{R}^{M}_{+}$ and presenting a formula for its subgradient. In particular, we show that the vector collecting the negative of the \emph{excess} rate of Type I error of the test $\varphi_{\kappa}$ in \eqref{eqn:test_NP_form} is a subgradient of $f(\cdot)$ at $\kappa$. 

\begin{lemma}\label{lem:f_convex}
The function $f(\kappa)$ defined in Equation \ref{eqn:NP_oracle} is convex. Furthermore, a subgradient of $f$ at $\kappa$ is given by
\[\nabla f(\kappa) \equiv -\left(\int \varphi_\kappa f_1 d \nu - \alpha,...,\int \varphi_\kappa f_M  d \nu - \alpha\right),\]
where $\varphi_{\kappa}$ is defined as in \eqref{eqn:test_NP_form}.
\end{lemma}

\begin{proof}
See Section \ref{subsection:proof_f_convex} of Appendix \ref{online_appendix}.
\end{proof}

Lemma \ref{lem:f_convex} shows that dual problem in \eqref{eq:dual} has, as expected, a convex objective function, and we present a simple formula for a subgradient. We now show that the dual problem can be further simplified by restricting the multipliers to belong to the bounded domain: 
\begin{equation}
\label{eq:bounded_domain}
    \mathcal{X} \equiv \left\{ \kappa \in \mathbb{R}_+^M: \| \kappa \|_1 \leq \frac{1}{\alpha} \right\}.
\end{equation}

\begin{lemma}
\label{prop:dual_bounded}
    $\inf_{\kappa \in \mathbb{R}_{+}^{M}} f(\kappa) = \inf_{\kappa \in \mathcal{X}} f(\kappa)$.
\end{lemma}

\begin{proof}
See Section \ref{subsection:proof_dual_bounded} in Appendix \ref{online_appendix}. 
\end{proof}
Lemma \ref{lem:f_convex} and \ref{prop:dual_bounded} show that in order to find the multipliers associated with the dual program in \eqref{eq:dual}, it is sufficient to solve a convex optimization problem over a bounded domain (in particular, an $\ell_{1}$-ball around the origin with radius $1/\alpha$) instead of solving a convex optimization problem over all of $\mathbb{R}^{M}_{+}$. An intuitive explanation of this result can be given as follows. Consider the hypothesis testing problem 
\[ \mathbf{H_0}: Y \sim f_0 \textrm{ vs. } \mathbf{H}_1: Y \sim f_1. \]
Suppose that $Y$ is real-valued and that both $f_0$ and $f_1$ are absolutely continuous with respect to Lebesgue measure. The Neyman-Pearson lemma implies that most powerful test of size $\alpha$ rejects if and only if 
\[ \frac{f_1}{f_0} > c_\alpha, \]
where $c_{\alpha}$ is the critical value that satisfies
\[ P_0 \left( \frac{f_1}{f_0} > c_{\alpha} \right) = \alpha. \]
Markov's inequality trivially implies that
\begin{eqnarray*}
\alpha=P_0 \left( \frac{f_1}{f_0} > c_{\alpha} \right) &\leq& c^{-1}_\alpha E_0 \left[ \frac{f_1}{f_0} \right]  \\
&=& c^{-1}_\alpha  \int \frac{f_1(y)}{f_0(y)} f_0(y)dy \\
&=& c^{-1}_\alpha.
\end{eqnarray*}
Thus $c_{\alpha} \leq 1/\alpha$. Lemma \ref{prop:dual_bounded} can be viewed as a generalization of this simple observation. In the next section we show that---due to Lemma \ref{lem:f_convex} and \ref{prop:dual_bounded}---the structure of the dual problem in \eqref{eq:dual} is amenable to the application of the methods of stochastic mirror descent. 

\section{Main Results} \label{sec:main_results} 
This section presents our main results. First, we present a formal definition of a stochastic mirror descent (S-MD) routine for convex optimization. The members of the mirror descent family are indexed by what is called a \emph{mirror map}. We apply S-MD to the problem in \eqref{eq:dual} by setting the mirror map to be equal to the negative entropy. This choice is motivated by our Lemma \ref{prop:dual_bounded} and the fact that negative entropy is a common recommendation for convex problems over $\ell_1$ balls; see Section 5.7 in \cite{nemirovski2009robust} and Example 2 in \cite{srebro2012convex}. Second, we define an \emph{approximate least-favorable distribution} and show that S-MD provably obtains an approximate least favorable distribution (Theorem \ref{thm:epsilon_least_favorable}). Finally, we define a \emph{nearly optimal test} and show that the S-MD routine can be used to generate such a test (Theorem \ref{thm:average_test}).

\subsection{Stochastic Mirror Descent} \label{subsection:SMD}
This section follows as closely as possible the notation in Sections 4.1 and 6.1 of \cite{bubeck2015convex}. Let $\mathbb{R}^{M}_{++}$ denote the set of all strictly positive vectors in $\mathbb{R}^M$. We say that a map $\Phi:\mathbb{R}_{++}^{M} \rightarrow \mathbb{R}$ is a mirror map if it satisfies the following properties
\begin{enumerate}
\item [i)] $\Phi$ is strictly convex and differentiable. 
\item [ii)] The gradient of $\Phi$ takes all possible values, that is $\nabla \Phi(\mathbb{R}_{++}^{M}) = \mathbb{R}^{M}$.
\item [iii)] The gradient of $\Phi$ diverges on the boundary of $\mathbb{R}^{M}_{++}$. 
\end{enumerate}
Mirror Maps are used to build iterative algorithms for constrained optimization problems when unbiased estimators of the gradient are available. More precisely, consider the optimization problem 
\[ \inf_{\kappa \in \mathcal{X}} f(\kappa),   \]
where $f:\mathcal{X} \rightarrow \mathbb{R}$ is a convex function and $\mathcal{X} \subseteq \mathbb{R}_{+}^{M}$. Suppose that $\widehat{G}(\kappa)$ is an unbiased estimator of a subgradient of $f$ at $\kappa$, in the sense that $\mathbb{E}\left[ \widehat{G}(\kappa) \right]$ is a subgradient of the function $f$ at $\kappa$.\footnote{The expectation should be understood as being conditional on $\kappa$, since $\kappa$ is stochastic. See Chapter 6 in \cite{bubeck2015convex}.} Let $D_{\Phi}(\kappa,\kappa')$ denote the Bregman divergence associated with $\Phi$, that is 
\[  D_{\Phi}(\kappa,\kappa') \equiv \Phi(\kappa) - \Phi(\kappa') - \nabla \Phi(\kappa') (\kappa-\kappa').  \]
The stochastic mirror descent algorithm (henceforth, S-MD) given the mirror map $\varphi$ is defined as follows: 

\begin{algorithm}[h!] 
\caption{Stochastic Mirror Descent with mirror map $\Phi$, stopped after $T$ epochs.}
\label{alg:stochastic_mirror_descent_general}
\begin{algorithmic}[1]
\STATE \textbf{Input:} Step size $\eta > 0$,  number of epochs $T \in \mathbb{N}$. 
\STATE Initialize $\kappa_1 \in \arg \min_{\kappa \in \mathcal{X} \cap \mathbb{R}^{M}_{++}} \Phi(\kappa) .$
\FOR{$t = 1, \ldots, T-1$ }
\STATE   
\begin{equation} \label{eqn:general-SMD-update}
    \kappa_{t+1} = \arg \min_{\kappa \in \mathcal{X} \cap \mathbb{R}^{M}_{++}} \eta \left(\widehat{G}(\kappa_{t})^{\top} \kappa \right) + D_{\Phi}(\kappa,\kappa_{t}).
\end{equation}
\ENDFOR
\STATE \textbf{Output:} $\bar{\kappa}_{T} \equiv \frac{1}{T} \sum_{t=1}^{T} \kappa_{t}$. 
\end{algorithmic}
\end{algorithm}
The general interpretation of the S-MD update in equation \eqref{eqn:general-SMD-update} is that \emph{``the method is trying to minimize the local linearization of the function while not moving too far away from the previous point, with distances measured via the Bregman divergence of the mirror map''}; see p. 301 in \cite{bubeck2015convex}. 

An important observation regarding Algorithm \ref{alg:stochastic_mirror_descent_general} is that its output is the \emph{average value} of $\kappa_{t}$ over all the iterations, and not its last value. Averaging the trajectories of a stochastic optimization routine is commonly referred to as Polyak-Ruppert averaging; see \cite{ruppert1988efficient} and \cite{polyak1992acceleration}. This idea goes back to seminal work on mirror descent by \cite{nemirovski1983problem}.

The following theorem specializes the mirror descent routine in Algorithm \ref{alg:stochastic_mirror_descent_general} to the dual problem in \ref{eq:dual}. The mirror map is set to be equal to the negative entropy. 

\begin{thm} \label{thm:Mullers_algo_is_SMD}
Consider the optimization problem $\inf_{\kappa \in \mathcal{X}} f(\kappa)$, where $f(\cdot)$ is defined in \eqref{eqn:NP_oracle} and the set $\mathcal{X}$ is defined in \eqref{eq:bounded_domain}. 
\begin{enumerate}
\item Let $\kappa_{t}$ be a realization of an  arbitrary $\mathcal{X}$-valued random vector. For each $m=1,\ldots, M$, let $Y_{m,1},..., Y_{m,N}$ be i.i.d. random variables with distribution $Y \sim f_m$ sampled independently of the realized value of $\kappa_{t}$. For any $N\geq 1$  
\begin{equation} \label{eq:gradient_Theorem1}
\widehat{G}_{N}(\kappa_{t}) \equiv -\left(\frac{1}{N}\sum_{n=1}^N \varphi_{\kappa_t}(Y_{1,n}) -\alpha, ..., \frac{1}{N}\sum_{n=1}^N \varphi_{\kappa_t}(Y_{M,n})- \alpha\right)^{\top} ,
\end{equation}
is an unbiased estimator of the subgradient of $f$ at $\kappa_{t}$; where $\varphi_{\kappa_t}$ is a test of the Neyman-Pearson form defined in \eqref{eqn:test_NP_form}. 
\item The stochastic mirror descent update in Algorithm \ref{alg:stochastic_mirror_descent_general} based on $\widehat{G}_{N}(\cdot)$ and the mirror map 
$\Phi(\kappa) = \sum_{m=1}^{M} \kappa_m \ln(\kappa_{m})$ is
\begin{equation} \label{eqn:SMD_update_Theorem1}
\kappa_{t+1,m} = c_t \cdot \kappa_{t,m}\exp \left( -\eta \cdot \widehat{G}_{m,N}(\kappa_{t})   \right),
\end{equation}
where $\widehat{G}_{m,N}(\kappa_{t})$ is the $m$-th coordinate of $\widehat{G}_{N}(\kappa_{t})$ and
 \begin{equation*}
c_t \equiv \min \left\{1,  \frac{1}{\alpha \sum_{m=1}^{M} \kappa_{t,m} \exp \left(- \eta \cdot \widehat{G}_{m,N}(\kappa_{t})   \right)}  \right\}.
\end{equation*}

\item The initial condition in Algorithm \ref{alg:stochastic_mirror_descent_general} based on the mirror map $\Phi(\kappa) = \sum_{m=1}^{M} \kappa_m \ln(\kappa_{m})$ is 
\begin{equation} \label{eq:initial_condition_theorem}
\kappa_{1} = \begin{cases}
    \left(\frac{1}{\exp(1)}, \ldots, \frac{1}{\exp(1)}\right), & \textrm{if } 1 \leq M < \frac{\exp(1)}{\alpha},\\
    \left(\frac{1}{\alpha M}, \ldots, \frac{1}{\alpha M } \right) & \textrm{if } M \geq \frac{\exp(1)}{\alpha}. 
\end{cases}
\end{equation}
\end{enumerate}
\end{thm}

\begin{proof}
See Section \ref{subsection:proof_Mullers_algo_is_SMD} in Appendix \ref{online_appendix}. 
\end{proof}

It is useful to make an explicit connection between the updating formula in \eqref{eqn:SMD_update_Theorem1} and Equation 10, p. 782 in \cite{elliott2015nearly}. Following the notation in \cite{elliott2015nearly}, define 
\[ \mu^{t+1}_{m} \equiv \ln(\kappa_{t+1,m}).  \]
Taking logarithms on both sides of \eqref{eqn:SMD_update_Theorem1} and using the definition of $\widehat{G}_{N}(\kappa_{t})$ yields
\begin{equation} \label{eqn:muller_update}
\mu^{t+1}_{m} = \ln(c_t) + \mu^{t}_{m} + \eta \left( \frac{1}{N}\sum_{n=1}^N \varphi_{\exp(\mu^{t})}(Y_{m,n}) -\alpha \right). 
\end{equation}
When $c_{t}=1$, the term $\ln(c_t)=0$, and thus, \eqref{eqn:muller_update} essentially matches Equation 10, p. 782 in \cite{elliott2015nearly} after noting that $\widehat{G}_{m,N}(\kappa_t)$ is a Monte Carlo estimate of the negative excess rate of Type I error the test $\varphi_{\kappa_{t}}$:
\[ \alpha-\int \varphi_{\kappa_t} f_m d \nu.   \]
Other than notation, the main difference between our expressions is the presence of the additional term $c_{t}$. In the stochastic mirror descent routine, this term is used to take into account the fact that---because of our Lemma \ref{prop:dual_bounded}---the optimization domain in the dual problem \eqref{eq:dual} can be restricted to values of $\kappa$ such that  $\sum_{m=1}^{M}\kappa_m \leq 1/\alpha$. 

\subsection{Approximate Solutions to the Dual Problem}
\label{subsection:sol_dual_prob}

In this subsection we show that if the number of epochs ($T$) and the step size ($\eta$) are chosen appropriately, then $\bar{\kappa}_{T} \equiv (1/T) \sum_{t=1}^{T} \kappa_{t}$---obtained using Equations \ref{eqn:SMD_update_Theorem1} and \ref{eq:initial_condition_theorem} in Theorem \ref{thm:Mullers_algo_is_SMD}---indeed can be used to generate an approximate least-favorable distribution for the problem in \eqref{eqn:testing_problem_intro}.  

In order to formalize this statement, note that we have defined $\bar{v}$ as the value of the dual problem $\inf_{\kappa \in \mathcal{X}} f(\kappa)$, where $f(\cdot)$ is defined in \eqref{eqn:NP_oracle} and the set $\mathcal{X}$ is defined in \eqref{eq:bounded_domain}. Let 
\begin{equation} \label{eqn:simplex_RM}
    \Delta^{M-1} \equiv \left\{ \lambda \in \mathbb{R}^{M} \: | \: \lambda_{m} \geq 0 \textrm{ for all } m=1, \ldots, M \quad \textrm{and} \quad \sum_{m=1}^{M} \lambda_{m} = 1 \right\}
\end{equation}
denote the probability simplex in $\mathbb{R}^{M}$. 

\begin{definition}[$\epsilon$-least favorable distribution] \label{def:epsilon-least-favorable} We say that a vector $\lambda_{\epsilon}^* \in \Delta^{M-1}$ is an \emph{$\epsilon$-least favorable distribution} if there exists a positive constant $\textrm{cv}_{\epsilon}^*$ such that 
\[ \kappa_{\epsilon}^* \equiv \textrm{cv}_{\epsilon}^* \cdot \lambda_{\epsilon}^*  \]
satisfies 
\begin{equation}
f(\kappa_{\epsilon}^*) \leq \bar{v} + \epsilon, 
\end{equation}
where $f(\cdot)$ is defined in \eqref{eqn:NP_oracle}. We say that a vector $\lambda^* \in \Delta^{M-1}$ is an \emph{approximate least-favorable distribution} if there exists $\epsilon>0$ for which $\lambda^*$ is $\epsilon$-least favorable. 
\end{definition}

\begin{remark} \label{remark:least_favorable}
The definition presented above is different from the notion of ``$\epsilon$-approximate least favorable distribution'' used by \cite{elliott2015nearly}. They define an $\epsilon$-approximate least favorable distribution as any distribution for which the corresponding Neyman-Pearson test has size exactly equal to $\alpha$, and has power at most  $\epsilon$ away from that of the most powerful test (see their  Definition 1, p. 780). In contrast, we  focus on the optimization problem that defines such a least-favorable distribution: the dual problem we presented in \eqref{eq:dual}. Our notion allows the associated Neyman-Pearson test to be over (or under) sized. This flexibility enables us to avoid the computational cost of having to simulate rates of Type I error for different critical values, and then to search over critical values either at the last round or within each iteration. In Section \ref{sec:nearly.optimal.test}, we show that with our definition, it is still possible to theoretically control the Type I error of our recommended nearly-optimal test (see our Theorem \ref{thm:average_test} and the simulation results in Section \ref{section:illustrative_example}). 

Our definition is inspired by a large literature in theoretical computer science, operations research, and optimization where it is a common approach to \emph{``to relax the requirement of finding an optimal solution, and instead settle for a solution that is good enough}''  \cite[p. 14]{shmoys2011design}. There are different criteria that can be used to formalize the statement that an approximate solution is ``good enough'', but a typical choice in optimization problems relies on its value function (which is the metric we use in our Definition \ref{def:epsilon-least-favorable}). We deliberately chose an additive approximation error, because most of the results that we are familiar with regarding the approximation error of mirror descent routines---and, more generally, first-order methods for convex optimization problems---take this form; see, for example, Section 5.1.1 in \cite{juditsky20115} and also \cite{bubeck2015convex}. But it is also possible to give results for multiplicative approximation errors; for example, see Theorem 1 in \cite{chen2017robust}.
\end{remark}

We now present a result showing that Algorithm \ref{alg:stochastic_mirror_descent_general} provably generates an approximate least-favorable distribution. For any nonnegative real number $x$, let $\lceil x \rceil$ denote the ``ceiling function''; that is smallest integer larger than $x$. 

\begin{thm} \label{thm:epsilon_least_favorable}
    Consider the optimization problem $\inf_{\kappa \in \mathcal{X}} f(\kappa)$, where $f(\cdot)$ is defined in \eqref{eqn:NP_oracle} and the set $\mathcal{X}$ is defined in \eqref{eq:bounded_domain}. Let $\overline{\kappa}_{T}$ be the output of Algorithm \ref{alg:stochastic_mirror_descent_general} based on the mirror map 
$\Phi(\kappa) = \sum_{m=1}^{M} \kappa_m \ln(\kappa_{m})$ and the unbiased estimator of the gradient $\widehat{G}_{N}(\cdot)$ defined in Equation \ref{eq:gradient_Theorem1} of Theorem \ref{thm:Mullers_algo_is_SMD}. If $\alpha\in(0,1/2)$ and 
$M > \exp(1)/\alpha$,
\begin{equation} \label{eqn:T_and_eta_Theorem2}
T = \left\lceil \frac{4(1 - \alpha)^2}{\alpha^2 \epsilon^2} \cdot \ln(M) \right\rceil, \quad \textrm{and} \quad  
\eta = \alpha \cdot \frac{\epsilon}{2(1 - \alpha)^2},
\end{equation}
then 
\begin{equation} \label{eqn:LFD_Theorem2}
\lambda^*_{T} \equiv \frac{\overline{\kappa}_{T}}{ \sum_{m=1}^{M} \overline{\kappa}_{T,m}}
\end{equation}
is a $\left( 1 + \frac{2\Omega}{\sqrt{\ln(M)N (1 - \alpha)^2}} \right)\epsilon $-least favorable distribution, in the sense of our Definition \ref{def:epsilon-least-favorable}, with probability at least $1-\exp\left( - \Omega^2 \right)$. 
\end{thm}

\begin{proof}
See Section \ref{subsection:Proof_Theorem_2} in Appendix \ref{online_appendix}.
\end{proof}

\noindent Theorem \ref{thm:epsilon_least_favorable} shows that---even after finitely many iterations---the S-MD routine for the dual problem, 
\[ \inf_{\kappa \in \mathcal{X}} f(\kappa),\] generates an approximate least-favorable distribution (in the sense of our Definition \ref{def:epsilon-least-favorable}) with high probability. The approximate least favorable distribution in \eqref{eqn:LFD_Theorem2} is the \emph{direction} of the multipliers 
$\overline{\kappa}_T$
(in analogy to what we would do if we had access to the exact solution of the dual problem). The probabilistic statement in the theorem arises due to the randomness in the gradient estimator in \eqref{eq:gradient_Theorem1}, which makes the output of the S-MD routine behave as a random variable. Note that if we fix the frequency at which we would like to obtain an approximate least-favorable distribution (over different \emph{runs} of the S-MD routine), then the number of draws used to construct the gradient estimator ($N$) determines how close we get to finding a ``good enough'' solution for the dual problem. 

For example, suppose that $\epsilon=.1$, $\alpha=10\%$ and $M=200$, and suppose we set $\Omega = \sqrt{\ln(1/\alpha)}$, so that $1-\exp(-\Omega^2) = 1-\alpha = 90\%$. If we run the S-MD routine using $N=1$ (only one draw per density $f_m$), then with probability 90\% we will obtain a 
\[  \left(1 + 2 \sqrt{\frac{\ln(1/\alpha)}{\ln(M)(1-\alpha)^2}} \right)\epsilon \approx 2.5 \epsilon = .25 \]
least-favorable distribution. If we use $N=10$ (only ten draws per density) we get a 
\[ \left(1 + 2 \sqrt{\frac{\ln(1/\alpha)}{\ln(M)\cdot 10 \cdot (1-\alpha)^2}} \right)\epsilon \approx 1.5 \epsilon = .15 \]
least favorable distribution. If we use $N=100$, with probability 90\% we get a $1.15\epsilon = .11$-least favorable distribution. More generally, Theorem \ref{thm:epsilon_least_favorable} shows that, for any target probability, we can always make $N$ large enough to get as close as we would like to the desired approximation error $\epsilon$. 

In our view, the most surprising part of Theorem \ref{thm:epsilon_least_favorable} is that even if the number of draws per density used to implement the S-MD routine are as low as $N=1$, it is still possible to get an approximate least-favorable distribution that provides a non-trivial approximate solution to the dual problem in \eqref{eq:dual}. For instance, in the example above, using only one draw per density yields an approximation error of $2.5 \epsilon =.25$ with probability 90\%. The approximation error of $.25$ should be interpreted as a worst-case guarantee that applies to any testing problem of the form \eqref{eqn:testing_problem_intro}. As we show in our illustrative example, the resulting approximation error can, in practice, be considerably smaller.

\subsection{Nearly Optimal Tests via Stochastic Mirror Descent}\label{sec:nearly.optimal.test} 

Now that we have established that the S-MD routine in Algorithm \ref{alg:stochastic_mirror_descent_general} (with negative entropy as a mirror map) provably generates an approximate least-favorable distribution---in the sense of our Definition \ref{def:epsilon-least-favorable}---we discuss the extent to which the S-MD routine can also be used to generate a \emph{nearly optimal test}. 

Before presenting a formal definition of what we mean by a nearly optimal test, it is helpful to explain why it is not entirely trivial to translate the approximate least-favorable distribution in %Theorem 
Equation \ref{eqn:LFD_Theorem2} into a nearly optimal test. Let $\kappa^*_{T}$ be the multipliers that we obtain after running the S-MD routine, with $T$ and $\eta$ defined as in Theorem \ref{thm:epsilon_least_favorable}. Consider the test of the Neyman-Pearson form, $\varphi_{\kappa^*_{T}}$, defined in \eqref{eqn:test_NP_form} based on the multipliers $\kappa^*_{T}$.\footnote{Note that $\kappa^*_{T}$ can be decomposed into its direction $\lambda^*_{T}$ as in Equation \ref{eqn:LFD_Theorem2} and its norm $\textrm{cv}^*_{T} \equiv \sum_{m=1}^{M}\kappa^*_{T,m}$. Thus, the test in \eqref{eqn:test_NP_form} rejects the null if and only if
\[ g(y) > \textrm{cv}^*_{T} \sum_{m=1}^{M} \lambda^*_{T,m} f_m(y).  \]
} Since the S-MD routine never explicitly tried to enforce \emph{size} control, it is possible that the size of $\varphi_{\kappa^*_{T}}$ is strictly above the nominal level $\alpha$.  Mathematically, this happens because an approximate optimizer to the dual problem does not necessarily imply a feasible solution to the primal problem (let alone a nearly optimal one). This suggests that, when defining a nearly optimal test, it could be helpful to take into account i) possible violations of the required size; ii) as well as potential loss in power, relative to the optimal solution.   

Let $\bar{v}$ be defined as value function of the dual problem in \eqref{eq:dual}. As we show in Section \ref{subsection:appendix_duality} of Appendix \ref{section:appendix_additional_results}, duality holds, and $\bar{v}$ equals the power of the most powerful test of size $\alpha$.\footnote{For the sake of exposition, we deliberately write $\bar{v}$ instead of $\bar{v}(\alpha)$.}  

\begin{definition} \label{def:nearly_optimal}
    A statistical test $\varphi^{\star}_{\epsilon,\delta}: \Yc \to [0,1]$ is said to be $(\epsilon, \delta)$-nearly optimal of size $\alpha$ if: 
    \begin{enumerate}
        \item The size of $\varphi^{\star}_{\epsilon,\delta}$ is no larger than $\alpha (1 + \delta)$, 
        \[
        \int \varphi^{\star}_{\epsilon,\delta} f_m d \nu \leq \alpha (1 + \delta), \quad \textrm{for all } m=1,\ldots, M. 
        \]
        \item The power of $\varphi^{\star}_{\epsilon,\delta}$ is at most $\epsilon$ away of the maximum power of a test of size $\alpha$, 
        \[
        \int \varphi^{\star}_{\epsilon,\delta} g d \nu \geq \bar{v} - \epsilon.
        \]
    \end{enumerate}
We say that test $\varphi^*$ is nearly optimal of size $\alpha$, if there exists $\epsilon,\delta>0$ for which $\varphi^*$ is $(\epsilon,\delta)$-nearly optimal.    
\end{definition}

The following theorem shows that Algorithm \ref{alg:stochastic_mirror_descent_general} can provably be used to generate a nearly optimal test. 

\begin{thm} \label{thm:average_test}

Consider the optimization problem $\inf_{\kappa \in \mathcal{X}} f(\kappa)$, where $f(\cdot)$ is defined in \eqref{eqn:NP_oracle} and the set $\mathcal{X}$ is defined in \eqref{eq:bounded_domain}. Let $\{\kappa_{t}\}_{t=1}^{T}$ be the sequence of multipliers generated by Algorithm \ref{alg:stochastic_mirror_descent_general} based on the mirror map 
$\Phi(\kappa) = \sum_{m=1}^{M} \kappa_m \ln(\kappa_{m})$ and the unbiased estimator of the subgradient $\widehat{G}_{N}(\cdot)$ defined in Equation \eqref{eq:gradient_Theorem1} of Theorem \ref{thm:Mullers_algo_is_SMD}. If $M > \exp(1)/\alpha$,
\begin{equation*}
T = \left\lceil \frac{4(1 - \alpha)^2}{\alpha^2 \epsilon^2} \cdot \ln(M) \right\rceil, \quad \textrm{and} \quad  
\eta = \alpha \cdot \frac{\epsilon}{2(1 - \alpha)^2},
\end{equation*} 
then the test
\begin{equation}
\label{eq:def_average_test}
    \bar{\varphi}_T(y) \equiv \frac{1}{T} \sum_{t=1}^{T} \varphi_{\kappa_t}(y)
\end{equation}
is nearly optimal of size $\alpha$ with high probability (where $\varphi_{\kappa_{t}}$ is the test of the Neyman-Pearson form in \eqref{eqn:test_NP_form}). More concretely, with probability at least $1 - \exp(- \Omega^2)$
    \begin{enumerate}
        \item For any $m = 1,\ldots, M$,     
    \begin{equation}\label{eq:average_test_t1error}
            \int \bar{\varphi}_{T} f_m d\nu  \leq \alpha \left(  1+\left[1 +  \frac{2 \Omega}{\sqrt{\ln(M) N (1 - \alpha)^2} } \right]\epsilon -  \frac{1}{T}\sum_{t=1}^{T}\widehat{G}_{N}(\kappa_{t})^{\top}  \kappa_{t}\right). 
        \end{equation} 
        \item The power of $\bar{\varphi}_{T}$ is larger than 
        \begin{equation} \label{eq:average_test_power}
        \bar{v} - \left[1 + \frac{ 2 \Omega}{\sqrt{\ln(M) N (1 - \alpha)^2}}\right]\epsilon.
        \end{equation}
    \end{enumerate}

\end{thm}
\begin{proof}
See Section \ref{subsection:Proof_Theorem_3} in Appendix \ref{online_appendix}. 
\end{proof}
\noindent Theorem \ref{thm:average_test} shows that the S-MD routine analyzed in this paper provably generates a nearly optimal test, in the sense of our Definition \ref{def:nearly_optimal}. There are three aspects about our result that are worth highlighting.  

First, it is rather surprising that the nearly optimal test in \eqref{eq:def_average_test} takes the form of an ``average'' test. In order to get some intuition of why this construction is helpful to obtain theoretical results, it is useful to explicitly write the dual in \eqref{eq:dual} as the \emph{minimax} problem 
\[
\min_{\kappa \in \mathbb{R}_{+}^{M}} \max_{\varphi} L(\varphi, \kappa),
\]
where $L(\varphi,\kappa)$ is the Lagrangian function defined in \eqref{eq:lagr_def}. In this problem, the ``min'' player is choosing a vector of (Lagrange) multipliers, and the ``max'' player is choosing a test. For a fixed $\kappa$, the \emph{best response} of the max player is a test of the Neyman-Pearson form $\varphi_{\kappa}$ defined in \eqref{eqn:test_NP_form}. A mirror descent routine for this problem initializes the choice of $\kappa$ by the ``min'' player, and iteratively updates its values based on the (sub)gradient of the Lagrangian with respect to $\kappa_t$, which---by results analogous to the envelope theorem---will give the rates of Type I error of $\varphi_{\kappa_{t}}$. The most powerful test of size $\alpha$ is the solution to the \emph{maximin} problem
\[
 \max_{\varphi} \min_{\kappa \in \mathbb{R}_{+}^{M}} L(\varphi, \kappa).
\]
The question is how to translate the iterates, $\{\kappa_{t}\}_{t=1}^{T}$, into a solution to the maximin problem. This question is common in the application of the mirror descent algorithm to minimax problems that arise in game theory and statistical decision theory; see \cite{fernandez2024epsilon} and the discussion of matrix games in \cite{arora2012multiplicative}. While these papers consider different problems to the one studied in this paper, a suggestion therein is to use the best responses of the max player $\{\varphi_{\kappa_{t}}\}_{t=1}^{T}$ and randomize over them with uniform probability. In our problem, such a construction becomes the average test in \eqref{eq:def_average_test}; and this is what motivated us to study its performance. To the best of our knowledge, our results have not been stated elsewhere.   

Second, the upper bound on the rate of Type I error features a term whose value changes with each specific run of the S-MD routine. This is not ideal, but we were not able to derive a better bound. To better understand the role of this term, consider again the example we discussed after Theorem \ref{thm:epsilon_least_favorable}. Suppose that $\epsilon=.1$, $\alpha=10\%$ and $M=200$, and suppose we set $\Omega = \sqrt{\ln(1/\alpha)}$, so that $1-\exp(-\Omega^2) = 1-\alpha = 90\%$. If we run the S-MD routine using $N=1$ (only one draw per density $f_m$),    
\[ 1 + \left[1 + 2\sqrt{\frac{\ln(1/\alpha)}{\ln(M)(1-\alpha)^2}} \right]\epsilon \approx 1+2.5 \epsilon = 1.25. \]
If the term 
\begin{equation} \label{eq:additional_term_size}
-\frac{1}{T}\sum_{t=1}^{T}\widehat{G}_{N}(\kappa_{t})^{\top}  \kappa_{t},
\end{equation}
were not part of  \eqref{eq:average_test_t1error}, then we could conclude that with probability at least 90\%, the test $\bar{\varphi}_{T}$ has size of at most $12.5\%$ (that is, there is a size distortion of 2.5\%) and power that is no less than $\bar{v}$ minus 25 percent points. Again, making $N$ arbitrarily large makes the size closer to $\alpha(1+\epsilon)$ and the power at least $\bar{v}-\epsilon$. The interpretation of Theorem \ref{thm:average_test} changes slightly when we incorporate \eqref{eq:additional_term_size}.  Suppose for example that in one run of the S-MD routine the term in \eqref{eq:additional_term_size} equals .05. Then, for that run, our best hope is that \eqref{eq:average_test_t1error} is satisfied with the larger bound $1.3$ instead of $1.25$. This means that we could see a rate or Type I error of the average test as high as 13\%. As we discuss in the next section, in our illustrative example the term in \eqref{eq:additional_term_size} tends to be small (and negative), but we do not have any theoretical guarantees for this.\footnote{The terms $-\widehat{G}_{N}(\kappa_{t})^{\top}  \kappa_{t}$ is a Monte-Carlo estimate of the average excess rate of Type I error of $\varphi_{\kappa_{t}}$, evaluated at the different null densities. The term in \eqref{eq:additional_term_size} averages these Monte-Carlo estimates over all iterations.}

Third, to report the entire test (as function of all possible data), one would have to retain the history of multipliers obtained from the S-MD routine. When both $M$ and the number of iterations of S-MD are large, this could come with significant computational and data storage expense. However, a typical use case is to perform the test on a specific data set. For this purpose, data storage requirements can be much reduced because, rather than retaining the weights for every epoch, it suffices to store the implied test results, i.e. one bit per epoch and parameter value being tested (and even less if one is content with only updating the average). Furthermore, rather than reporting a rejection probability, it usually suffices to either accept or reject, although in cases where the rejection probability is interior, this decision will then be random conditionally on the data. But this can be achieved at much lower computational expense: (ii) By a simple application of the Law of Iterated Expectations, rejecting with probability corresponding to the average test is equivalent to first drawing a ``realized epoch'' $t^* \sim \operatorname{unif}(\{1,\ldots,T\})$ and then executing the Neyman-Pearson test for epoch $t^*$ only. (iii) As $t^*$ can be drawn before starting the iteration, it is only necessary to execute iterations up to epoch $t^*$, threreby only executing $T/2$ iterations in expectation. By using these simplifications, it should be easy to execute the test.

An alternative to the average test is to simply report the Neyman-Pearson test in \eqref{eqn:test_NP_form} associated with the multipliers $\bar{\kappa}_{T}$ obtained as the output of Algorithm \ref{alg:stochastic_mirror_descent_general} based on the mirror map 
$\Phi(\kappa) = \sum_{m=1}^{M} \kappa_m \ln(\kappa_{m})$ and the unbiased estimator of the gradient $\widehat{G}_{N}(\cdot)$ defined in Equation \ref{eq:gradient_Theorem1} of Theorem \ref{thm:Mullers_algo_is_SMD}. While it is challenging to analyze the size and power properties of this test for finite $T$, in Appendix \ref{subsubsection:test_average_kappa} we show that, under some conditions, as $T \rightarrow \infty$ the power of the test will converge to $\bar{v}$, and it will have correct size (in a sense we make precise).  

\subsection{Remarks on Confidence Regions}

It is common to construct confidence regions by inverting tests. However, the test advocated here is in general randomized, raising the question of how to invert it. Conceptual discussions of this matter go back at least to the 1950's \citep{Stevens50,Lehmann59}.\footnote{A notable difference to our setting is that these discussions were motivated by the randomized nature of optimal or exact tests in highly discrete sample spaces.} To summarize some key points, for this paragraph only let the test function $\rho(\cdot)$ also depend on the parameter value to be tested, i.e. we temporarily define $\rho:\mathcal{Y} \times \Theta \to [0,1]$, where $\rho(y,\theta)$ is the probability of rejecting the instance of $H_0$ characterized by parameter value $\theta$ given data $y$; \citet{Lehmann59} calls this the \emph{critical function}. Then we can define no less than four intervals that arguably invert our test:
\begin{enumerate}
    \item \citet{Lehmann59} defines a randomized confidence region as the set $\{\theta:\rho(y,\theta)\leq u\},$ where $u$ is a realization of $U \sim \operatorname{unif}(0,1)$, reflecting data-independent randomization by the statistician.
    \item \citet{GM05} propose to directly report $1-\rho(y,\theta)$ as function of $\theta$ and to interpret it as membership function of a fuzzy set; it is easy to see that expected membership of the true parameter value will correspond to the target coverage.
    \item Similarly to our discussion just above, one could draw a random epoch $t^*$ and invert the corresponding (nonrandomized) Neyman-Pearson test.
    \item Again similar to previous discussion, one could invert the Neyman-Pearson test that utilizes average weights $\bar{\kappa}_{T}$.
\end{enumerate}
Mirroring discussions in the previous subsection, idea 4 is the computationally most involved and its justification is asymptotic, idea 3 will be the computationally easiest, and idea 1 is intermediate; 2 is computationally equivalent to 1. Some users might find 2 hard to interpret \citep{CasellaBerger05}. We leave further analysis of the issue to future research. 

\section {Illustrative Example} \label{section:illustrative_example}

In order to illustrate the performance of the SM-D routine in Algorithm \ref{alg:stochastic_mirror_descent_general}---along with the implications of the theoretical guarantees in Theorem \ref{thm:epsilon_least_favorable} and Theorem \ref{thm:average_test}---we consider an elementary testing problem that arises in the context of the univariate Gaussian location model. More precisely, suppose we observe a realization of the random variable
\[ Y \sim \mathcal{N}(\theta,1). \]
The location parameter, $\theta$, is unknown to the econometrician. Let $\Theta_0 \equiv \{\theta_{0,1}, \ldots, \theta_{0,M}\}$ be an equally-spaced grid over the interval $[-5,0]$ consisting of $M=200$ points. We order the elements of $\Theta_0$ in decreasing order, so that $\theta_{0,1} = 0$. We also define the singleton set $\Theta_1 \equiv \{\theta_1\}$. 

We assume that the econometrician is interested in the following hypothesis testing problem:
\[ \textbf{H}_0: \theta \in \Theta_0 \quad \textrm{vs.} \quad \textbf{H}_1: \theta \in \Theta_1.  \]
It is well known that the most powerful test of size $\alpha$ for this problem---which we denote as $\varphi^*_{\alpha}$---rejects the null if $Y$ is large enough. More precisely, $\varphi^*_{\alpha}(Y) \equiv \mathbf{1}\{ Y \geq z_{1-\alpha} \}$, where $z_{1-\alpha}$ is the $1-\alpha$ quantile of a standard normal. The power of this test can then be expressed in terms of the normal c.d.f. as $\Pr(N(0,1) \leq \theta_1 - z_{1-\alpha})$. Since the most powerful test of size $\alpha$ for this example is known, we can use this information to analyze the theoretical guarantees we provided in Theorem \ref{thm:epsilon_least_favorable} and Theorem \ref{thm:average_test}.

\emph{The Stochastic Mirror Descent (S-MD) Routine:} We first obtain a nearly optimal test of size $\alpha = 10\%$. We set $\theta_1 = 2$, which means that the largest power of a test of size $\alpha=10\%$ is $\Phi(2-1.28) \approx 76.38\%$. This is also the value of the dual problem in \eqref{eq:dual}. We set $\epsilon = .1$, and use the formulae in Theorem \ref{thm:epsilon_least_favorable} to determine the maximum number of iterations ($T$) and the learning rate ($\eta$) for the S-MD routine:
\begin{equation} \label{eqn:T_and_eta_example}
T = \left\lceil \frac{4(1 - \alpha)^2}{\alpha^2 \epsilon^2} \cdot \ln(M) \right\rceil = 171,666, \quad 
\eta = \alpha \cdot \frac{\epsilon}{2(1 - \alpha)^2} = .0062.
\end{equation}

In this example $M = 200 > (\exp\left(1\right)/\alpha) = \exp(1)\cdot10$. Thus, in accordance with our theoretical derivations, the initial condition for the S-MD routine ($\kappa_0 \in \mathbb{R}^{M}$) is chosen to be:
\[ \kappa_0 = \left( 1/(\alpha M) , \ldots,  1/(\alpha M) \right)^{\top} = (.05, \ldots, .05)^{\top}. \]

The main component of the S-MD routine is the stochastic mirror descent update. We implement the unbiased estimator of the gradient in Theorem \ref{thm:Mullers_algo_is_SMD} using only one draw per density; that is, $N=1$. More precisely, if we let $f_{m}(\cdot)$ denote the p.d.f. of $Y$ under the null hypothesis $\theta_{0,m} \in \Theta_0$ and we let $g(\cdot)$ be the p.d.f. of $Y$ under the alternative $\Theta_1$, the mirror descent update necessitates an unbiased estimator of the rates of Type I error of the test 
\begin{equation} \label{eqn:NP_test_NormalExample} 
\varphi_{\kappa}(Y) \equiv \mathbf{1} \left\{ g(Y) > \sum_{m=1}^{M} \kappa_{m} f_m(Y) \right \}. 
\end{equation}
In our example, an unbiased estimator for the rate of Type I error at $\theta_{0,m}$ can be succinctly obtained by sampling $Z \sim N(0,1)$ and using $\varphi_{\kappa}(Z + \theta_{0,m})$ as an estimator. More precisely, in each epoch $t$ we obtain one draw  $Z_t \sim \mathcal{N}(0,1)$ and compute the mirror descent update (coordinate by coordinate) as
\[
\kappa_{m,t+1} \equiv \kappa_{m,t} \cdot \exp\left( \eta \left[ \varphi_{\kappa_t}(Z_t + \theta_{0,m}) - \alpha \right] \right), \quad \text{for each } m=1, \ldots , M. 
\]

The intuition of the update is very simple. If $\varphi_{\kappa}(Z + \theta_{0,m})$---the unbiased estimator of the rate of Type I error of $\varphi_{\kappa_{t}}$ at $\theta_{0,m}$---is larger than $\alpha$, then  $\kappa_{m,t+1}$ increases (and otherwise decreases). We also know that $\| \kappa_{t} \|_{1}$ must be less than or equal than $1/\alpha$. Thus, after the update we check if $\| \kappa_{t+1} \| \leq 1/\alpha$. If this is the case, we keep $\kappa_{t+1}$ as is; but otherwise we normalize $\kappa_{t+1}$ to guarantee that $\sum_{m=1}^{M} \kappa_{m} \leq 1/\alpha$. This gives us back the update described in part 2 of Theorem \ref{thm:Mullers_algo_is_SMD}.

In general, updating $\kappa_{t}$ is numerically very cheap when $N=1$, as it only involves obtaining one sample from each of the null densities (along with the evaluation of the null and alternative densities at each draw) and also evaluation of the exponential function. Using Matlab R2024a on a personal ASUS Vivobook Pro 15 @ 2.5GHz Intel Core Ultra 9 185H, it took around 280 seconds (slightly less than 5 minutes) to complete all of the $T=171,666$ iterations. 

\emph{Approximate Least-Favorable Distribution:} As suggested by our Theorem \ref{thm:epsilon_least_favorable}, in order to construct an approximate least-favorable distribution we standardize the average value of $\kappa_t$ to represent it as a probability distribution. More precisely, the blue bars in Figure \ref{fig:kappas_normal_testing} below correspond to 
\[ \lambda_{T} \equiv \frac{ \bar{\kappa}_{T} }{ \| \bar{\kappa}_{T}\|_1 }, \quad \textrm{where } \bar{\kappa}_{T} \equiv \frac{1}{T} \sum_{t=1}^{T} \kappa_{t}. \]
In the testing problem we are considering, it is known that the least-favorable distribution loads all of its mass on $\theta_{0,1}$. As Figure \ref{fig:kappas_normal_testing} shows, the output of the S-MD routine resembles such distribution.

\begin{figure}%[h!]
    \centering
    \includegraphics[width=0.6\textwidth]{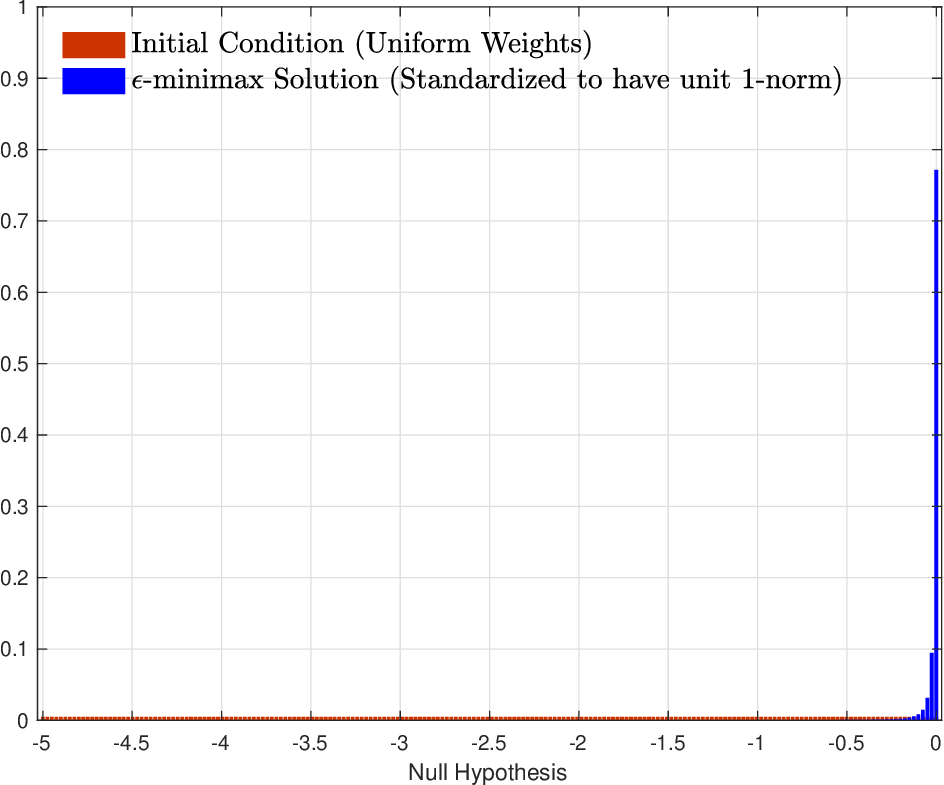}
    \caption{$\lambda_{T} \equiv \bar{\kappa}_{T} / \| \bar{\kappa}_{T}\|_1$ for $\alpha = 10\%$ and $\epsilon = 0.1$; where $\bar{\kappa}_{T}= (1/T) \sum_{t=1}^{T} \kappa_{t}.$ } 
    \label{fig:kappas_normal_testing}
\end{figure}

Our definition of approximate least-favorable distribution in Definition \ref{def:epsilon-least-favorable} makes reference to the value function of the dual problem in \eqref{eq:dual}, and not in terms of its minimizer. In this example, it is easy to show that the value of the dual (which we denoted by $\bar{v}$) equals $76.38\%$ (the power of the most powerful test of size $\alpha=10\%$). We now argue that, as expected, the distribution $\lambda_{T}$ approximately solves the dual problem, in the sense of Definition \ref{def:epsilon-least-favorable}. To see this, we just need to evaluate the function:
\[ f(\bar{\kappa}_{T}) = \underbrace{\int \varphi_{\bar{\kappa}_{T}} g(y)dy}_{\approx 76.96\%}- \underbrace{\sum_{m=1}^{M} \bar{\kappa}_{T,m} \left( \int \varphi_{\bar{\kappa}_{T}}(y) f_m(y) dy  - \alpha \right)}_{-1.01\%} \approx 77.97\%,  \]
where $\phi_{\bar{\kappa}_{T}}$ is the test of Neyman-Pearson form defined in \eqref{eqn:test_NP_form}, and where a Monte-Carlo approximation with 100,000 draws is used for the evaluation of each of the integrals. Thus, in this example:
\[f(\bar{\kappa}_{T}) \approx 77.97\% < \bar{v} + \epsilon = 76.38\% + 10\% = 86.38\%.  \]
This means that in our run of the S-MD routine we obtained an $\epsilon=.016$-least-favorable distribution. The quality of the approximation is much better than what we expected based on our theoretical results in Theorem \ref{thm:epsilon_least_favorable}. This is consistent with the fact that the results in Theorem \ref{thm:epsilon_least_favorable} apply to every possible testing problem with $M$ null densities and a single alternative. We also conducted a Monte-Carlo simulation where we implemented the S-MD routine with different draws for the estimation of the subgradient, with 10,000 draws being used for integral evaluation in each. In all of the 100 runs we obtain a $10\%$-least favorable distribution. This is consistent that the theoretical results we presented in Theorem \ref{thm:epsilon_least_favorable} apply to any testing problem of the form $\eqref{eqn:testing_problem_intro}$.

\emph{Nearly Optimal Test $\bar{\varphi}_T$:} Figure \ref{fig:phi_bar} reports the test $\bar{\varphi}_T$ (red, solid line), which is the test defined in \eqref{eq:def_average_test}. For comparison, we also report the test $\varphi_{\bar{\kappa}_T}$ (blue, solid line). The size of $\bar{\varphi}_{T}$ is approximately $10.22\%$, and its power is approximately $76.98\%$. This means that the test $\bar{\varphi}_{T}$ is slightly over-sized, but it has competitive power.  

\begin{figure}%[h!]
    \centering
    \includegraphics[width=0.6\textwidth]{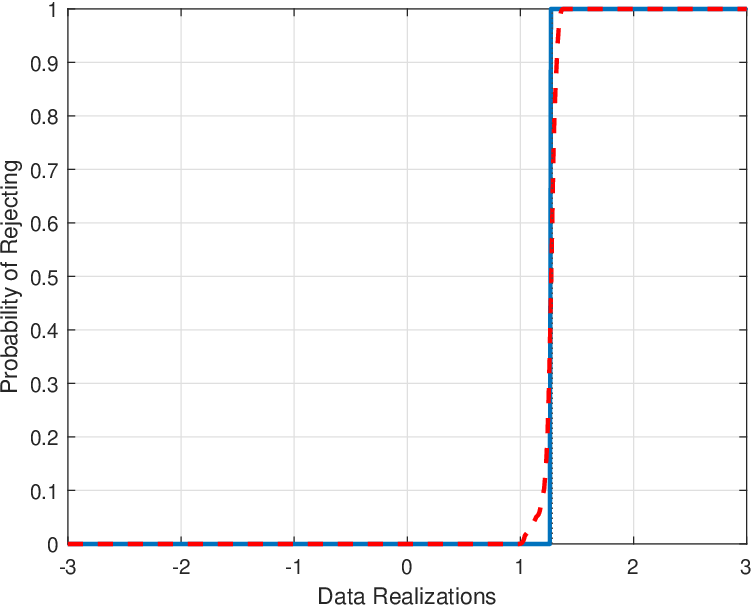}
    \caption{$\bar{\varphi}_T$ (red) alongside $\varphi_{\bar{\kappa}_T}$ (blue) for $\alpha=10\%$ and $\epsilon=.1$.} 
    \label{fig:phi_bar} 
\end{figure}

We also conducted 100 different runs of our S-MD routine. In all of the 100 hundred runs the size was at most $\alpha(1+\epsilon)$ and the lowest power achieved was 76.82\%. 

\emph{Time Comparison of Using More Draws in the S-MD Routine:} We also analyzed the increased computational effort of increasing the number of draws used to evaluate the subgradient during each iteration. To do this, we re-ran our illustrative example using 1, 10, 100, and 1,000 draws in each round. Table \ref{table:z_draws_runtimes} shows the runtime as a function of the number of draws. Lastly, we also calculated the expected runtime associated with 20,000 draws in each round---the number of draws recommended by \cite{elliott2015nearly}.\footnote{To be clear, they recommended to have a fixed set of draws at the beginning of the iterations and to reuse those draws via importance sampling at each iteration. While this may be computationally more feasible for certain distributions, most theoretical results on S-MD require the draws in each period to be independent of the history (otherwise, it is difficult  to guarantee unbiasedness of the estimator of the subgradient). The theoretical performance guarantee with their recommendations may be derived with additional assumptions, but remains unknown  to the best of our knowledge.} Our expected runtime was 369,697 seconds.\footnote{This estimate was generated by taking the time difference between the first and second round, and then multiplying by $T$. Note that this estimate is smaller than extrapolation from the last two observations based on a linear trend.} We think these results illustrate the computational gains of implementing the S-MD routine with a small number of draws to evaluate the subgradient. 

\begin{table}
    \centering
    \begin{tabular}{|c|c|}
        \hline
        Draws & Runtime (seconds) \\
        \hline
        1 & 314\\
        \hline 
        10 & 882\\
        \hline 
        100 & 4,028\\
        \hline
        1,000 & 36,320\\
        \hline
    \end{tabular}
    \caption{Number of draws in S-MD routine during each round versus runtime.}
    \label{table:z_draws_runtimes}
\end{table}

\section{Approximately Unbiased Tests} \label{sec:extensions}

Consider now a variation of the testing problem in \eqref{eqn:testing_problem_intro} where the alternative hypothesis is also composite, but with only $I$ possible distributions for the data:
\begin{equation} \label{eqn:testing_problem_composite}
\mathbf{H}_0: \textrm{ the density of $Y$ is $f_{m}$, \: $m =1, \ldots, M,$} \quad \textrm{ \emph{vs.} } \quad \mathbf{H}_1: \textrm{ the density of $Y$ is $g_i$, \: $i=1,\ldots, I$.} 
\end{equation}
As explained in Section 2.2 of \cite{elliott2015nearly}, one can reduce the problem in \eqref{eqn:testing_problem_composite} to the problem in \eqref{eqn:testing_problem_intro} by choosing weights $w \equiv (w_1, \ldots, w_{I}) \in \Delta^{I-1}$ and defining $g \equiv \sum_{i=1}^{I} w_i g_i$. Then, the test $\varphi$ that solves  \eqref{eq:main_problem} can be interpreted as the test that maximizes $w$-weighted average power among all tests of size at most $\alpha$. 

A common criticism of tests that maximize a weighted average power criterion (henceforth, WAP) is that they can be \emph{biased}: their power for some density $g_i$ can be lower than $\alpha$\citep{moreira2013contributions,andrews2016conditional}. \cite{moreira2013contributions} note that one could include additional constraints in the problem \eqref{eqn:testing_problem_intro} and consider: 
\begin{equation}\label{eq:main_problem_extension}
    \sup_{\varphi:\mathcal{Y} \rightarrow [0,1]} \int \varphi g d\nu, \quad \text{s.t.} \quad \int \varphi f_{m}d\nu \leq \alpha, \hspace{1em} m = 1,...,M, \quad \int \varphi g_{i} d \nu \geq \alpha, \hspace{1em} i=1,\ldots, I. 
\end{equation}
Just as before, we can define the Lagrangian function associated with problem \eqref{eq:main_problem_extension} as 
\begin{equation}
    \label{eq:lagr_def_extension}
        L(\varphi, \kappa,\mu) \equiv \int \varphi g d\nu -  \sum_{m=1}^M \kappa_m \left[\int \varphi f_{m}d\nu - \alpha\right] - \sum_{i=1}^{I} \mu_{i}\left[\int \varphi (-g_{i})d\nu + \alpha\right],
\end{equation}
where we refer to $\kappa \equiv (\kappa_1,...,\kappa_M) \in \mathbb{R}_{+}^{M}$ as the Lagrange multipliers associated with each of the inequality constraints that bound the test's rate of Type I error, and we let $\mu \equiv (\mu_1,...,\mu_I) \in \mathbb{R}_{+}^{I}$ denote the Lagrange multipliers associated with each of the inequality constraints preventing the test to be biased.   

We could then proceed as we did before and define the dual optimization problem:
\begin{equation}
\label{eq:dual_extension}
    \inf_{\kappa \in \mathbb{R}_{+}^{M},\: \mu \in \mathbb{R}^{I}_{+}} f(\kappa,\mu),
\end{equation}
where
\begin{equation}\label{eqn:NP_oracle_extension}     
     f(\kappa,\mu) \equiv \sup_{\varphi:\mathcal{Y} \rightarrow [0,1]} L(\varphi,\kappa,\mu). 
\end{equation}
It is possible to show that the function $f(\kappa,\mu)$ is convex in its arguments. Moreover, a test $\varphi$ that achieves the maximum is the test 
\begin{equation} \label{eqn:test_NP_form_extension} 
\varphi_{\kappa,\mu}(y) \equiv 
\left\{
\begin{array}{cc}
1 & \text{if } g(y) > \sum_{m=1}^{M} \kappa_m f_m(y) - \sum_{i=1}^{I} \mu_i g_i(y),  \\
0 & \text{if } g(y) \leq \sum_{m=1}^{M} \kappa_m f_m(y)- \sum_{i=1}^{I} \mu_i g_i(y), 
\end{array}
\right.
\end{equation}
and a subgradient of $f(\cdot)$ at $(\kappa,\mu)$ is
\[\nabla f(\kappa,\mu) \equiv -\left(\int \varphi_{\kappa,\mu} f_1 d \nu - \alpha,...,\int \varphi_{\kappa,\mu} f_M \nu d \nu - \alpha, \int \varphi_{\kappa,\mu} (-g_{1})d\nu + \alpha \ldots \int \varphi_{\kappa,\mu} (-g_{I})d\nu + \alpha \right).\]

If $\nabla f(\kappa,\mu)$ were known, the mirror descent routine (with negative entropy as mirror map) for this problem would have updates 
\begin{eqnarray*} \label{eqn:SMD_update_extension}
\kappa_{t+1,m} &=& \kappa_{t,m}\exp \left( -\eta \cdot \nabla_m f(\kappa_{t},\mu_{t}) \right), \\
\mu_{t+1,i} &= & \mu_{t,i}\exp \left( -\eta \cdot \nabla_i f(\kappa_{t},\mu_{t}) \right).
\end{eqnarray*}
Establishing a result similar to Theorem \ref{thm:epsilon_least_favorable} and \ref{thm:average_test} is more challenging because the application of standard results would need ex-ante constraints on the $\|\cdot\|_{1}$-norm of $\kappa$ and $\mu$. But, for example, if one knew that the optimal values of $\kappa$ and $\mu$ satisfied the constraint $\| \kappa + \mu \|_{1} < 1/\alpha$ then our previous theoretical results for S-MD would apply.

\section{Conclusion} \label{sec:conclusion}
We showed that---in testing problems where the null hypothesis postulates $M$ distributions for the observed data---one can use a stochastic mirror descent routine to \emph{provably} obtain---after finitely many iterations---both an approximate least-favorable distribution and a nearly optimal test. The convex program that arises naturally in the testing problem in  \eqref{eqn:testing_problem_intro} is the \emph{dual} of the mathematical program that defines the \emph{most powerful test} of \emph{size $\alpha$}.

Our theoretical results allowed us to provide concrete recommendations about the algorithm's implementation: including its initial condition, its step size, the number of iterations, and the number of stochastic draws per iteration that can be used to approximate the subgradient of the objective function. These practical recommendations have at least two important implications. First, the number of iterations used by the algorithm scales logarithmically in $M$ (which means there is no theoretical sense in which the algorithm scales poorly as a function of the elements in the null hypothesis). Second, the algorithm can be implemented with a single stochastic draw per null density in each iteration (taking a larger number of draws improves the approximation error of the S-MD routine, but using a small number of draws reduces the computational burden of the algorithm). Importantly, our suggested algorithm coincides with a slight variation of the algorithm in \cite{elliott2015nearly}.

\bibliography{references}
\bibliographystyle{ecta}

\begin{appendix}
\section{Proofs of Main Results}\label{online_appendix}

\subsection{Proof of Lemma~\ref{lem:f_convex}}
\label{subsection:proof_f_convex}

\begin{proof}
To prove convexity, note that by definition
\begin{align*}
f(\lambda \kappa + (1-\lambda) \kappa') &= \sup_\varphi \int \varphi g d \nu - \sum_{m=1}^M (\lambda \kappa_m + (1-\lambda) \kappa_m') \left[\int \varphi f_{m} d \nu - \alpha\right],
\end{align*}
where we have slightly abused notation by omitting the fact that $\varphi$ is allowed to be an arbitrary element of the space of all randomized tests. 
Consequently, 
\begin{align*}
f(\lambda \kappa + (1-\lambda) \kappa') &\leq \lambda \sup_{\varphi}\left\{\varphi g d\nu - \sum_{m=1}^M \kappa_m \left[\int \varphi f_{m} d \nu - \alpha\right]\right\} \\
&+ (1-\lambda) \sup_\varphi \left\{\int \varphi g d \nu - \sum_{m=1}^M \kappa_m' \left[\int \varphi f_{m} d \nu -\alpha\right]\right\}\\
&= \lambda f(\kappa) + (1-\lambda)f(\kappa').
\end{align*}
Therefore, $f$ is convex in $\kappa$.

To show that $\nabla f(\kappa)$ is a subgradient of $f$ at $\kappa \in \mathbb{R}_+^M$, we need to show that for any $\kappa' \in \mathbb{R}_{+}^{M}$
\begin{align*}
    f(\kappa) \leq f(\kappa') + \nabla f(\kappa)(\kappa - \kappa').
\end{align*}
Note first that the test $\varphi_{\kappa}$ solves the problem 
\begin{equation} \label{eq:aux_lemma1}
\sup_{\varphi:\mathcal{Y}\rightarrow [0,1]} \int \varphi g d \nu - \sum_{m=1}^M \kappa_m \left[\int \varphi f_m d \nu-\alpha\right]. 
\end{equation}
We can then rewrite \eqref{eq:aux_lemma1} as
\[\sup_{\varphi:\mathcal{Y}\rightarrow [0,1]} \int \varphi \left( g -  \sum_{m=1}^M \kappa_m f_m \right) d \nu + \alpha \sum_{m=1}^M \kappa_m. \]
Consequently, 
\begin{align*}
    f(\kappa) &=  \int \varphi_\kappa g d\nu - \sum_{m=1}^M \kappa_m\left[\int \varphi_\kappa f_m  d\nu-\alpha\right]\\
    &= \int \varphi_\kappa g d\nu - \sum_{m=1}^M (\kappa_m-\kappa_m')\left[\int \varphi_\kappa f_m d \nu-\alpha\right]\\
    &- \sum_{m=1}^M \kappa_m' \left[\int \varphi_\kappa f_m d \nu - \alpha\right]\\
    &= \int \varphi_k g d\nu - \sum_{m=1}^M \kappa_m'\left[\int \varphi_\kappa f_m d\nu - \alpha\right] + \nabla f(\kappa) (\kappa - \kappa')\\
    & \leq f(\kappa') + \nabla f(\kappa)(\kappa - \kappa').
\end{align*}
\end{proof}

\subsection{Proof of Lemma \ref{prop:dual_bounded}} \label{subsection:proof_dual_bounded}

\begin{proof}
Since $\mathcal{X} \subset \mathbb{R}_+^{M}$ we have
\[ \inf_{\kappa \in \mathbb{R}_{+}^{M}} f(\kappa) \leq \inf_{\kappa \in \mathcal{X}} f(\kappa). \]
Thus, it is sufficient to show that
\[ \inf_{\kappa \in \mathbb{R}_{+}^{M}} f(\kappa) \geq \inf_{\kappa \in \mathcal{X}} f(\kappa). \]
Suppose this is not the case and that 
\[ \bar{v} \equiv \inf_{\kappa \in \mathbb{R}_{+}^{M}} f(\kappa) < \inf_{\kappa \in \mathcal{X}} f(\kappa). \]
By definition of infimum, for any $\epsilon>0$ there exists $\kappa_{\epsilon} \in \mathbb{R}^{M}_{+}$ such that 
\[ \bar{v} \leq f(\kappa_{\epsilon}) < \bar{v} + \epsilon.\] 
By choosing $\epsilon$ small enough, we can guarantee the existence of an element $\kappa_{\epsilon} \in \mathbb{R}_{+}^{M}$ such that 
\[ f(\kappa_{\epsilon}) <  \bar{v} + \epsilon < \inf_{\kappa \in \mathcal{X}} f(\kappa).  \]
Since $\mathcal{X}$ and $\left(\mathbb{R}^{M}_{+} \backslash \mathcal{X} \right)$ form a partition of $\mathbb{R}^{M}_{+}$, it must be the case that either $\kappa_{\epsilon} \in \mathcal{X}$ or $\kappa_{\epsilon} \in \left(\mathbb{R}^{M}_{+} \backslash \mathcal{X} \right)$. Clearly, we cannot have $\kappa_{\epsilon} \in \mathcal{X}$ (as this would immediately yield a contradiction). Thus, we must have $\kappa_{\epsilon} \in \left(\mathbb{R}^{M}_{+} \backslash \mathcal{X} \right)$.

Note that at $\kappa = \mathbf{0}$, 
    \[
    f(\mathbf{0}) = \int \varphi_{\mathbf{0}} g d\nu \leq \int g d\nu = 1.  
    \]
    But also, for any $\kappa$ such that $\|\kappa\|_1 > 1/\alpha$, 
    \[
    f(\kappa) = \sup_{\varphi:\mathcal{Y} \rightarrow [0,1]} \int \varphi \left[ g - \sum_{m=1}^M \kappa_m f_{m}\right] d\nu + \alpha \sum_{m=1}^M \kappa_m \geq \sum_{m=1}^{M} \kappa_m \alpha > 1.
    \]
Therefore, 
\[ 1< f(\kappa_{\epsilon}) < \inf_{\kappa \in \mathcal{X}} f(\kappa) \leq f(\mathbf{0}) \leq 1. \]
This yields a contradiction. We conclude that $\bar{v} \equiv \inf_{\kappa \in \mathbb{R}_{+}^{M}} f(\kappa) = \inf_{\kappa \in \mathcal{X}} f(\kappa).$
\end{proof}

\subsection{Proof of Theorem \ref{thm:Mullers_algo_is_SMD}} \label{subsection:proof_Mullers_algo_is_SMD} 

{\scshape Proof of Part 1 of the Theorem:} Let $\kappa_{t}$ be a realization of an arbitrary $\mathcal{X}$-valued random vector. Let $\widehat{G}_{m,N}(\kappa_t)$ be the $m$-th coordinate of $\widehat{G}_{N}(\kappa_{t})$. If suffices to show that if $(Y_{m,1}, \ldots, Y_{m,N})$ are i.i.d. random variables with distribution $Y \sim f_m$ sampled independently of the realized value of $\kappa_{t}$, then $\mathbb{E}[\widehat{G}_{m,N}(\kappa_t) | \kappa_t] = -\left(\int \varphi_\kappa f_m d\nu - \alpha \right)$. 

By definition of $\widehat{G}_{m,N}(\kappa_t)$, 
\begin{eqnarray*}
\mathbb{E}[\widehat{G}_{m,N}(\kappa_{t}) | \kappa_{t}] &=& -\mathbb{E}\left[ \frac{1}{N}\sum_{n=1}^N \varphi_{\kappa_t}(Y_{m,n}) \Big | \kappa_{t} \right] + \alpha \\
&=& -\mathbb{E}\left[ \varphi_{\kappa_t}(Y_{m,n}) \Big | \kappa_{t} \right] + \alpha,
\end{eqnarray*}
where the last line follows from the fact that $(Y_{m,1}, \ldots, Y_{m,N})$ are i.i.d. according to $f_m$, independently of the value of $\kappa_{t}$. Since $f_m$ is the p.d.f. of $Y$ relative to the $\sigma$-finite measure $\nu$, Problem 1, Chapter 5, p. 177 of \cite{dudley02} implies 
\[ \mathbb{E}\left[ \varphi_{\kappa_t}(Y_{m,n}) \Big | \kappa_{t} \right] = \int \varphi_{\kappa_{t}} f_m d \nu. \]
Therefore, $\widehat{G}_N(\kappa_t)$ is an unbiased estimator of the subgradient of $f$ at the realized $\kappa_t$.

{\scshape Proof of Part 2 of the Theorem:} Let $\widehat{G}_{N}(\kappa_{t})$ be the unbiased estimator of the subgradient of $f$ at $\kappa_{t}$. We provide an explicit solution for the problem 
\begin{align}\label{eq:stoch_mirror_statement}
    \kappa_{t+1} = \arg \min_{\kappa \in \mathcal{X} \cap \mathbb{R}^{M}_{++}} \eta \widehat{G}_{N}(\kappa_{t})^{\top}\kappa + D_\Phi(\kappa,\kappa_{t}),
\end{align}
when $\Phi(\kappa) = \sum_{m=1}^{M} \kappa_{m} \ln(\kappa_{m})$. By definition of Bregman divergence,
\[D_\Phi(\kappa,\kappa_t) = \Phi(\kappa) - \Phi(\kappa_t) - \langle \nabla \Phi(\kappa_t), \kappa - \kappa_t \rangle.\]
Consequently, $\kappa_{t+1}$ is the solution to the following optimization problem
\begin{equation}
\min_{\kappa \in \mathbb{R}_{++}^{M}} \eta \widehat{G}_{N}(\kappa_{t})^{\top}\kappa +  \sum_{m=1}^{M} \kappa_{m} \ln \left( \kappa_{m} / \kappa_{t,m} \right) -  \sum_{m=1}^{M} (\kappa_{m}-\kappa_{t,m}), 
\end{equation}
subject to the constraint 
\[ \sum_{m=1}^{M} \kappa_{m} \leq 1/\alpha. \]
Let $\mu$ denote the Lagrange multiplier associated with this constraint. Thus, the first-order conditions of the problem for each $\kappa_{t+1,m}$ become:
\begin{equation} \label{eqn:FOC_iterate}
\eta \widehat{G}_{m,N}(\kappa_{t}) + \ln(\kappa_{t+1,m}/\kappa_{t,m})  + \mu = 0,  
\end{equation}
where $\widehat{G}_{m,N}(\kappa_{t})$ is the m-th entry of $\widehat{G}_{N}(\kappa_{t})$.
The first-order condition in  \eqref{eqn:FOC_iterate} can be written as
\[ \kappa_{t+1,m} =  \kappa_{t,m} \exp \left(-\eta \widehat{G}_{m,N}(\kappa_{t}) \right) \exp \left(-\mu \right).   \]
Two cases to consider. First, if 
\[ \sum_{m=1}^{M} \kappa_{t,m} \exp \left( -\eta \widehat{G}_{m,N}(\kappa_{t}) \right) < 1/\alpha,   \]
then $\mu=0$ and 
\begin{equation} \label{eqn:aux_1_Theorem1_kappa} 
    \kappa_{t+1,m} =  \kappa_{t,m} \exp \left(-\eta \widehat{G}_{m,N}(\kappa_{t}) \right).
\end{equation}
Second, if 
\[ \sum_{m=1}^{M} \kappa_{t,m} \exp \left( -\eta \widehat{G}_{m,N}(\kappa_{t}) \right) \geq 1/\alpha,   \]
then $\mu>0$ and $\sum_{m=1}^{M} \kappa_{t+1,m}$ must equal $1/\alpha$. Consequently,
\begin{equation} \label{eqn:aux_1_Theorem1_kappa} 
    \kappa_{t+1,m} = \frac{1}{\alpha} \cdot \frac{ \kappa_{t,m} \exp \left(-\eta \widehat{G}_{m,N}(\kappa_{t}) \right)}{\sum_{m=1}^{M} \kappa_{t,m} \exp \left(-\eta \widehat{G}_{m,N}(\kappa_{t}) \right) },
\end{equation}
which can be achieved by setting
\[ \mu = \ln \left( \alpha \sum_{m=1}^{M} \kappa_{t,m} \exp \left(-\eta \widehat{G}_{m,N}(\kappa_{t}) \right)  \right). \]

\noindent {\scshape Proof of Part 3 of the Theorem:} The initial condition $\kappa_1$ solves:
\begin{equation}\label{eq:initial_kappa}
    \min_{\kappa \in \mathbb{R}^{M}_{++}}\sum_{m=1}^M \kappa_m \ln (\kappa_m) \quad \textrm{ s.t.} \quad \| \kappa \|_{1} = \sum_{m=1}^{M} \kappa_{m} \leq 1/\alpha. 
\end{equation}
We re-parameterize this problem by defining 
\[ K \equiv \| \kappa \|_{1}, \quad p_{m} \equiv \kappa_{m}/K, \quad p = (p_1, \ldots, p_{M})^{\top}.\]
Since $\kappa \in \mathbb{R}^{M}_{++}$, then $K>0$ and $w_{m}>0$ for all $m=1,\ldots, M$. Moreover, if we denote by $\Delta^{M-1}$ the simplex in $\mathbb{R}^{M}$ and use \textrm{int}$\left(\Delta^{M-1}\right)$ to denote its interior, the optimization problem in \eqref{eq:initial_kappa} thus becomes the nested optimization problem
\begin{equation}\label{eq:initial_kappa_aux}
    \min_{K > 0} \left(  \min_{p \in \textrm{int}\left(\Delta^{M-1}\right) } K \left( \sum_{m=1}^M p_m \ln (p_m) \right) +  K\ln(K) \right) \quad \textrm{ s.t.} \quad K \leq 1/\alpha. 
\end{equation}

Thus, we first solve the inner problem which consists of finding the distribution in the simplex with the smallest negative entropy:  
\[\min_{p \in \textrm{int}\left(\Delta^{M-1}\right) } \sum_{m=1}^M p_m \ln (p_m).\]
It is known that the solution of this problem is to set $p_{m}=1/M$. We verify this below for the sake of exposition. The first order conditions are
\[ 1+ \ln(p_{m}) + \mu = 0,\] 
where $\mu$ is the Lagrange multiplier associated with $\|p\|_1 = 1$. Solving for $p_m$ yields
\[ p_m = \exp \left( - (1+\mu) \right) \]
which implies (by summing the left side over $m=1,\ldots, M$):
\[ \frac{1}{M} = \exp \left( - (1+\mu) \right).  \]\
We conclude that $p^*_{m} = 1/M$ is the optimal  direction of $\kappa_1$. We now find its scale by solving the outer optimization problem
\begin{equation} \label{eq:initial_kappa_aux_scale}
\min_{K>0}  K \left( \sum_{m=1}^{M} \frac{1}{M} \ln\left( \frac{1}{M} \right) \right) + K\ln(K) = \min_{K>0}  K(\ln(K) - \ln(M)) \quad \textrm{s.t.} \quad K \leq 1/\alpha.    
\end{equation}
Without the constraint, the objective function has a global minimum at $K^*$ satisfying 
\[ \ln(K^*) + 1 - \ln(M) = 0, \]
or equivalently, $K^*= M / \exp(1)$. It is also decreasing for $K<K^*$ and increasing for larger values. Therefore, the solution to the problem in \eqref{eq:initial_kappa_aux_scale} is
\[K^* = \begin{cases}
    \frac{M}{\exp(1)} & \text{if } 1 \leq M < \frac{\exp(1)}{\alpha}\\
    \frac{1}{\alpha} & \text{if } M \geq \frac{\exp(1)}{\alpha}. 
\end{cases},\]
Thus, the initial condition is
\[\kappa_{1} = \begin{cases}
    \left(\frac{1}{\exp(1)}, \ldots, \frac{1}{\exp(1)}\right) & \text{if } 1 \leq M < \frac{\exp(1)}{\alpha},\\
    \left(\frac{1}{M\alpha}, \ldots, \frac{1}{M\alpha} \right) & \text{if } M > \frac{\exp(1)}{\alpha}. 
\end{cases}\]

\subsection{Proof of Theorem \ref{thm:epsilon_least_favorable}} \label{subsection:Proof_Theorem_2}

The approximation error of the numerical iteration consists of two parts: optimization error and estimation error. The former is intrinsic to the optimization algorithm when applying the exact (sub)gradient, while  latter is induced by  the  estimation error of the unknown subgradient. 

\begin{proof}
By Lemma \ref{lem:f_convex}, the function $f(\cdot)$ in the dual problem \eqref{eq:dual} is convex. Consequently, for any $\kappa \in \Xc$,
\begin{equation} \label{eqn:aux1_theorem_2}
f\left( \frac{1}{T} \sum_{t=1}^{T} \kappa_{t} \right) - f(\kappa) \leq \frac{1}{T} \sum_{t=1}^{T}f(\kappa_{t}) - f(\kappa). 
\end{equation}
%where $\kappa^*$ is the solution to the dual problem \eqref{eq:dual}.
Note under the S-MD routine of Algorithm \ref{alg:stochastic_mirror_descent_general}, $\kappa_{t}$ is a random variable. Part 1 of Theorem \ref{thm:Mullers_algo_is_SMD} showed that, given the realized value of $\kappa_{t}$,  $\widehat{G}_{N}(\kappa_{t})$ is an unbiased estimator of the subgradient of $f$ at $\kappa_{t}$; that is, $\mathbb{E}\left[ \widehat{G}_{N}(\kappa_{t}) \right] = \nabla f(\kappa_{t})^{\top}$. Consequently, Equation \eqref{eqn:aux1_theorem_2} and the definition of subgradient imply
\begin{eqnarray}
f\left( \frac{1}{T} \sum_{t=1}^{T} \kappa_{t} \right) - f(\kappa) &\leq& \frac{1}{T} \sum_{t=1}^{T} \nabla f(\kappa_{t}) \left( \kappa_{t}- \kappa \right) \nonumber \\
&=& \frac{1}{T} \sum_{t=1}^{T} \left( \nabla f(\kappa_{t})^{\top} - \widehat{G}_{N}(\kappa_{t})  \right)^{\top} \left( \kappa_{t}- \kappa \right) \label{eqn:aux2_theorem_3} \\
&+& \frac{1}{T} \sum_{t=1}^{T} \widehat{G}_{N}(\kappa_{t})^{\top} \left( \kappa_{t}- \kappa \right). \label{eqn:aux2_theorem_2} 
\end{eqnarray}

It follows by \citet[][proof of Theorem 4.2 and Equation (10) on p.307]{bubeck2015convex}  that \eqref{eqn:aux2_theorem_2} is bounded above by
\[
\frac{D_{\Phi}(\kappa, \kappa_1)}{\eta T} + \frac{\eta}{2\rho}\frac{1}{T} \sum_{t=1}^{T} \| \widehat{G}_{N}(\kappa_{t}) \|^2_{\infty},
\]
where $\rho$ is the parameter of the convexity of $\Phi(\cdot)$ with respect to $\|\cdot\|_1$. We have already proved in Lemma~\ref{lem:phi_convexity} that $\rho = \alpha/2$. Additionally, $D_{\varphi}(\kappa, \kappa_1) \leq \ln(M)/\alpha$.  \footnote{When $M > \frac{e}{\alpha}$, meaning the problem is high in dimension, then
\begin{align*}
R^2 & \equiv \sup_{\kappa \in \Xc} \Phi(\kappa) - \Phi(\kappa_1) \\
&= \frac{1}{\alpha} \ln\left(\frac{1}{\alpha}\right) - \frac{1}{\alpha} \left(\ln\left(\frac{1}{\alpha} \right) - \ln (M)\right)= \frac{1}{\alpha} \ln(M).
\end{align*}} 
Accordingly, \eqref{eqn:aux2_theorem_2} is bounded above by 

\begin{equation}
\frac{D_{\varphi}(\kappa, \kappa_1)}{\eta T} + \frac{\eta}{2\rho}\frac{1}{T} \sum_{t=1}^{T} \| \widehat{G}_{N}(\kappa_{t}) \|^2_{\infty}\leq \frac{\ln(M)}{\alpha T \eta} + \frac{\eta (1-\alpha)^2 }{\alpha},
\end{equation} 
where $\| \cdot \|_{\infty}$ is the sup norm, and where the last inequality follows from the fact that $\alpha<1/2$. Since 
\[ T = \left\lceil \frac{4(1 - \alpha)^2}{\alpha^2 \epsilon^2} \cdot \ln(M) \right\rceil, \quad \textrm{and} \quad  
\eta = \alpha \cdot \frac{\epsilon}{2(1 - \alpha)^2}, \] 
we conclude that \eqref{eqn:aux2_theorem_2} is at most $\epsilon$.  Next, we upper bound the term \eqref{eqn:aux2_theorem_3}. Define 
\begin{equation}
\label{eq:def_delta_t}
    \Delta_t \equiv \widehat{G}_{N}(\kappa_{t}) - \nabla f(\kappa_t).
\end{equation}
Given $t$ and $\kappa_t$, we write $\Delta_t$ as an average of $N$ independent vectors, i.e., 
\[ 
\begin{aligned}
    & \Delta_t = \frac{1}{N} \sum_{n=1}^{N} \Delta_{t,n},\\
    \text{ where } &
\Delta_{t,n} \equiv \left( \varphi_{\kappa_t}(Y^{(t)}_{1,n}) - \int \varphi_{\kappa_t} f_1 d\nu, ..., \varphi_{\kappa_t}(Y^{(t)}_{M,n}) - \int \varphi_{\kappa_t} f_M d\nu \right), n= 1,2,...,N.
\end{aligned}
\]
Denote the $M\times N$ random vectors at time $t$ as $Y_t \equiv (Y^{(t)}_{m,n})$, and let $\mathcal{F}_t = \sigma(Y_1,Y_2,..., Y_t)$ denote the canonical filtration of $Y_t$. From our iteration, $\kappa_{t}$ is $\mathcal{F}_t$-predictable, i.e., $\sigma(\kappa_{t}) \subset \mathcal{F}_{t-1}$ for each $t$.  Also note that  $\|\kappa_t - \kappa^*\|_1 \leq 2/\alpha$ for any $\kappa\in\mathcal{X}$. Thus, applying Lemma~\ref{lem:stoch_err_bound} with $X_t = \kappa_t - \kappa$ and  $L = 2/\alpha$,  we conclude that,  for a given confidence level $\Omega > 0$, \eqref{eqn:aux2_theorem_3} is upper bounded by 
\[
\frac{4 \Omega}{\alpha \sqrt{TN}} = \frac{2\Omega \epsilon}{\sqrt{(1- \alpha)^2\ln(M) N}},
\]
with probability at least $1 - \exp(- \Omega^2)$. The conclusion follows from 
combining the upper bound for \eqref{eqn:aux2_theorem_2} and \eqref{eqn:aux2_theorem_3}, and take $\kappa = \kappa^*$, the solution to the dual problem.
\end{proof}

\subsection{Proof of Theorem \ref{thm:average_test}}
\label{subsection:Proof_Theorem_3}

\begin{proof}
    First, it is already derived in the proof for Theorem~\ref{thm:epsilon_least_favorable} that for any $\kappa \in \Xc$,
    \begin{equation}
    \label{eqn:aux3_theorem_1}
    \frac{1}{T} \sum_{t=1}^{T} \widehat{G}_{N}(\kappa_{t})^{\top} \left( \kappa_{t}- \kappa \right) \leq \frac{\ln(M)}{\alpha \eta T} + \frac{\eta (1 - \alpha)^2}{\alpha} = \epsilon
    \end{equation}

    For a given $\kappa \in \Xc$, apply Lemma~\ref{lem:stoch_err_bound} with $X_t = \kappa$ and $X_t = \kappa_t$, respectively. Notice that both $\|\kappa\|_1 \leq 1/\alpha, \|\kappa_t\|_1 \leq 1/\alpha$ hold, then
    \begin{align*}
    & \Pr\left[\frac{1}{T} \sum_{t=1}^{T} (\widehat{G}_{N}(\kappa_{t}) - \nabla f(\kappa_t))^{\top}  \kappa < \frac{2 \Omega}{\alpha \sqrt{TN}}\right]  \geq 1- \exp(- \Omega^2),\\
    \text{and}\  & \Pr\left[\frac{1}{T} \sum_{t=1}^{T} -(\widehat{G}_{N}(\kappa_{t}) - \nabla f(\kappa_t))^{\top}  \kappa_{t} < \frac{2 \Omega}{\alpha \sqrt{TN}}\right] \geq 1 - \exp(- \Omega^2).
    \end{align*}
    Combining with \eqref{eqn:aux3_theorem_1}, we have
    \begin{equation}
    \label{eq:high_prob_t1error}
     \Pr\left[\frac{1}{T} \sum_{t=1}^{T} \widehat{G}_{N}(\kappa_{t})^{\top}  \kappa_{t} -
     \frac{1}{T} \sum_{t=1}^{T}\nabla f(\kappa_t)^{\top}  \kappa  < \epsilon + \frac{2 \Omega}{\alpha \sqrt{TN}}\right]  \geq 1- \exp(- \Omega^2),
    \end{equation}
    and similarly,
    \begin{equation}
    \label{eq:high_prob_power}
    \Pr\left[\frac{1}{T} \sum_{t=1}^{T} \nabla f(\kappa_t)^{\top}  \kappa_{t} - \frac{1}{T} \sum_{t=1}^{T} \widehat{G}_{N}(\kappa_{t})^{\top}  \kappa < \epsilon + \frac{2 \Omega}{\alpha \sqrt{TN}}\right]  \geq 1- \exp(- \Omega^2).
    \end{equation}

    Note \eqref{eq:high_prob_t1error} implies a high probability bound for the Type I error:  take $\kappa_m = 1/\alpha$ for $m = j$ and $\kappa_m = 0$ for other $m \neq j$.
    Then, 
    \begin{align*}
    \frac{1}{T} \sum_{t=1}^T \nabla f(\kappa_t)^{\top}  \kappa & = \frac{1}{T} \sum_{t=1}^T \sum_{m=1}^{M} \kappa_m (\alpha - \int \varphi_t f_m d\nu ) \\
    & = \alpha \cdot \frac{1}{\alpha} - \frac{1}{T}\frac{1}{\alpha}  \sum_{t=1}^T \int \varphi_t f_m d\nu = 1 - \frac{1}{\alpha}\int \bar{\varphi} f_m d\nu.
    \end{align*}
    Taking it back to \eqref{eq:high_prob_t1error}, we have
    \[
    \int \bar{\varphi} f_m d\nu  \leq \alpha \left( 1 - \frac{1}{T} \sum_{t=1}^{T} \widehat{G}_{N}(\kappa_{t})^{\top}  \kappa_{t} + \epsilon + \frac{2\Omega}{\alpha\sqrt{TN}}\right)
    \]
    with probability at least $1 - \exp(-\Omega^2)$.
    The first statement follows from the arbitrariness of  $m$.

   For the second statement, note according to \eqref{eq:high_prob_power},  we have
    \[
    \frac{1}{T} \sum_{t=1}^{T} \nabla f(\kappa_t)^{\top}  \kappa_{t} -  \epsilon - \frac{2 \Omega}{\alpha \sqrt{TN}}  \leq \frac{1}{T} \sum_{t=1}^{T} \widehat{G}_{N}(\kappa_{t})^{\top} \kappa 
    \]
    with probability at least $1 - \exp(- \Omega^2)$. Add the power of $\bar{\varphi}$, $\int \bar{\varphi} g d\nu$, to both sides. 
    Notice that 
    \[
    \power(\bar{\varphi}) + \frac{1}{T} \sum_{t=1}^{T} \nabla f(\kappa_t)^{\top}  \kappa_{t} = \frac{1}{T}\sum_{t=1}^{T} f(\kappa_t) \geq \bar{v},
    \]
    taking $\kappa = 0$ leads to
    \[
    \power(\bar{\varphi}) \geq \bar{v} - \epsilon - \frac{2 \Omega}{\alpha \sqrt{TN}},
    \]
    with probability at least $1-\exp(- \Omega^2)$. This concludes the proof of the second statement of Theorem~\ref{thm:average_test}. 
\end{proof}

\newpage

\global\long\def\thepage{OA-\arabic{page}}%
\setcounter{page}{1}

\section{Online Appendix} \label{section:appendix_additional_results} 

\subsection{Additional Lemmas} \label{subsection:appendix_additional_lem}

\begin{lemma}\label{lem:phi_mirror_map}
The function $\Phi:\mathbb{R}^{M}_{++} \rightarrow \mathbb{R}$ given by $\Phi (\kappa) = \sum_{m=1}^M \kappa_m \ln(\kappa_m)$ is a mirror map.
\end{lemma}
\begin{proof}
It is sufficient to verify the conditions i)-ii)-iii) given at the beginning of Section \eqref{subsection:SMD}, which are taken from Section 4.1 in \cite{bubeck2015convex}. We first verify i); namely that $\Phi(\cdot)$ is differentiable and strictly convex. The gradient of $\Phi$ at any $\kappa \in \mathbb{R}_{++}^{M}$ is:
\[\nabla \Phi(\kappa) = (1 + \ln(\kappa_1),...,1 + \ln(\kappa_M)).\]
Thus, $\Phi$ is differentiable in its domain. Moreover, for any $\kappa \in \mathbb{R}^{M}_{++}$ the Hessian takes the form
\[\begin{bmatrix}
    \frac{1}{\kappa_1} & 0 & ... & 0\\0 & \frac{1}{\kappa_2} & ... & ...\\
    ... & ... & ... & ...\\
    0 & ... & ... & \frac{1}{\kappa_M}
\end{bmatrix},\]
which is a positive definite matrix. Therefore, $\Phi(\cdot)$ is strictly convex in its domain. 

Next, we verify ii); namely, that $\nabla \Phi(\mathbb{R}^{M}_{++}) = \mathbb{R}^M$. Since $\ln(\mathbb{R}_{++}) = \mathbb{R}$, condition ii) holds.

Lastly, we verify condition iii), which in this case is equivalent to showing that for any $\kappa^*$ with one or more entries equal to zero satisfies
\[\lim_{\kappa \rightarrow  \kappa^*} ||\nabla \Phi(\mathbb{R}_{++}^M)|| = \infty.\]
Note that, if the $m$-th entry of $\kappa^*$ is zero, then
\[\lim_{x \rightarrow 0+} 1 + \ln(x)= -\infty.\]
Therefore, condition iii) holds. Therefore, $\Phi(\cdot)$ is a mirror map.
\end{proof}

\begin{lemma}
\label{lem:phi_convexity}
The function $\Phi (\kappa) = \sum_{m=1}^M \kappa_m \ln(\kappa_m)$ restricted on $\Xc$ is $\frac{\alpha}{2}$-strongly convex w.r.t. $\|\cdot\|_1$.
\end{lemma}

\begin{proof}
We intend to prove, for any $\kappa_1, \kappa_2 \in \Xc$, 
\begin{equation}
\label{eq:phi_convex}
    \Phi(\kappa_1) - \Phi(\kappa_2) - \langle \nabla \Phi(\kappa_2), \kappa_1 - \kappa_2 \rangle \geq \frac{\alpha}{4} \|\kappa_1- \kappa_2 \|_1^2
\end{equation}

Define $K_1 \equiv \|\kappa_1\|_1, K_2 \equiv \|\kappa_2\|_1$, and $p_1 = \kappa_1/ \|\kappa_1\|_1, p_2 = \kappa_2/ \|\kappa_2\|_1$, we write
\[
\kappa_1 = K_1 p_1, \kappa_2 = K_2 p_2. 
\]
Then, we decompose the left-hand side of \eqref{eq:phi_convex} as 
\[
\Phi(\kappa_1) - \Phi(\kappa_2) - \langle \nabla \Phi(\kappa_2), \kappa_1 - \kappa_2 \rangle = K_1 \ln \frac{K_1}{K_2} - (K_1 - K_2) + K_1\sum_{m=1}^{M} p_{1,m} \ln \frac{p_{1,m}}{p_{2,m}}
\]
Notice that:\\ 
1. the function $\Phi: (0, \frac{1}{\alpha}] \to \R$ defined by $\Phi(x) = x \ln(x)$ is $\alpha$-strongly convex, so 
\[
K_1 \ln \frac{K_1}{K_2} - (K_1 - K_2) \geq \frac{\alpha}{2} |K_1 - K_2|^2.
\]
2. $p_1, p_2$ are on the $(M-1)$-dimensional simplex. We can apply the Pinsker's inequality, 
\[
K_1\sum_{m=1}^{M} p_{1,m} \ln \frac{p_{1,m}}{p_{2,m}} \geq \frac{K_1}{2} \|p_1 - p_2 \|_1^2 \geq \frac{\alpha}{2} \|K_1(p_1 - p_2) \|_1^2.
\]
Together we get 
\begin{align*}
    \Phi(\kappa_1) - \Phi(\kappa_2) - \langle \nabla \Phi(\kappa_2), \kappa_1 - \kappa_2 \rangle & \geq \frac{\alpha}{2} |K_1 - K_2|^2 + \frac{\alpha}{2} \|K_1(p_1 - p_2) \|_1^2 \\
    & = \frac{\alpha}{2} \left( \| K_1 p_2 -K_2 p_2\|_1^2 + \|K_1 p_1 - K_1 p_2 \|_1^2 \right) \\
    & \geq \frac{\alpha}{4} \left( \| K_1 p_2 -K_2 p_2\|_1 + \|K_1 p_1 - K_1 p_2 \|_1\right)^2  \\
    & \geq \frac{\alpha}{4} \| K_1 p_1 -K_2 p_2 \|_1^2.
\end{align*}
\end{proof}

\begin{lemma}
\label{lem:stoch_err_bound}
    Suppose our unbiased estimator $\widehat{G}_{N}(\kappa_{t})$ is evaluated on $M \times N$ independent draws in $Y_t$. Let $\Delta_{t}$ be defined as in \eqref{eq:def_delta_t}. Then, 
     for any $\gamma > 0$ and any $\{ X_t \in \R^M, t = 1,2,...\}$ that is $\mathcal{F}_t$-predictable,
     if there exists a constant $L$ such that $\|X_t\|_1 \leq L$, 
     we have 
    \[
    \E\left[\exp\left(\frac{\gamma}{T} \sum_{t=1}^T \langle X_t, \Delta_t \rangle \right) \right] \leq \exp\left(\frac{\gamma^2 L^2 }{TN}\right),
    \]
    which leads to
    \[
    \Pr\left[\frac{1}{T}\sum_{t=1}^T \langle X_t, \Delta_t \rangle  \geq \delta\right] \leq \exp\left(- \frac{TN \delta^2 }{4L^2}\right).
    \]
Moreover,  for a given confidence level $\Omega>0$, we have 
 \[
    \Pr\left[\frac{1}{T}\sum_{t=1}^T \langle X_t, \Delta_t \rangle \geq \frac{2L \Omega}{\sqrt{TN}}\right] \leq \exp( - \Omega^2).
    \]
\end{lemma}
\begin{proof}
By the definition of $\Delta_{t}$ in \eqref{eq:def_delta_t} we have
\[ 
\begin{aligned}
    & \Delta_t = \frac{1}{N} \sum_{n=1}^{N} \Delta_{t,n},\\
    \text{ where } &
\Delta_{t,n} \equiv \left( \varphi_{\kappa_t}(Y^{(t)}_{1,n}) - \int \varphi_{\kappa_t} f_1 d\nu, ..., \varphi_{\kappa_t}(Y^{(t)}_{M,n}) - \int \varphi_{\kappa_t} f_M d\nu \right), n= 1,2,...,N.
\end{aligned}
\]
Note first that for any $t,n$, $\|\Delta_{t,n}\|_{\infty}\leq 1$. Since, by assumption, $\|X_t\|_1 \leq L$, we have
    \[
    \E \left[\exp \left(\frac{\langle X_t, \Delta_{t,n}\rangle^2}{L^2} \right) \Big |\mathcal{F}_{t-1} \right]\leq \exp(1).
    \]
    Apply the same steps as in \cite{nemirovski2009robust}. Consider first the case in which $0 < \gamma L \leq 1$. In this case, we apply $e^x \leq x + e^{x^2}$, 
\[
\begin{array}{ll}
\mathbb{E}\big[\exp\left(\gamma \langle X_t, \Delta_{t,n}\rangle\right)|\mathcal{F}_{t-1}\big] \leq \mathbb{E}\big[\exp\left(\gamma^2 \langle X_t, \Delta_{t,n}\rangle^2\right)|\mathcal{F}_{t-1}\big]  
\leq \exp\left(\gamma^2L^2\right),
\end{array}
\]
where the last inequality follows from the fact that $\|\Delta_{t,n}\|_{\infty}\leq 1$ and $\|X_t\|_1 \leq L$. 
Consider now the case in which $\gamma  L > 1$. Note that
\[
\mathbb{E}\big[\exp\left(\gamma\langle X_{t},\Delta_{t,n}\rangle\right)|\mathcal{F}_{t-1}\big]
%\leq	\mathbb{E}\left[\exp\left(\gamma\left\Vert X_{t}\right\Vert _{1}\left\Vert \Delta_{t,n}\right\Vert _{\infty}\right)|\mathcal{F}_{t-1}\right]
\leq\mathbb{E}\left[\exp\left(\gamma L\right)\big|\mathcal{F}_{t-1}\right]\leq\exp\left(\gamma^{2}L^{2}\right).
\]
Therefore, in both cases, 
    \[
    \E[\exp(\gamma \langle X_t, \Delta_{t,n} \rangle) |\mathcal{F}_{t-1}] \leq \exp( \gamma^2 L^2).
    \]
    Because $\{\Delta_{t,n}, n=1,2,..,N\}$ are independent, we have 
\begin{equation} \label{eq:aux1_Lemma_5}
    \E[\exp(\gamma \langle X_t, \Delta_{t} \rangle) |\mathcal{F}_{t-1}] = \Pi_{n=1}^{N} \E \left[\exp\left(\frac{\gamma}{N}\langle X_t, \Delta_{t,n} \rangle \right) \Big |\mathcal{F}_{t-1} \right]  \leq \exp\left( \frac{\gamma^2 L^2}{N}\right).
\end{equation}
Applying Law of Iterated Expectations sequentially yields
    \[
    \begin{aligned}
        \E\left[\exp\left(\frac{\gamma}{T} \sum_{t=1}^T \langle X_t, \Delta_t \rangle \right) \right] & = \E\left[\E \left[\exp\left(\frac{\gamma}{T} \sum_{t=1}^T \langle X_t, \Delta_t \rangle \right) \Big |\mathcal{F}_{T-1} \right]  \right]\\
    & = \E\left[ \E\left[\exp \left( \frac{\gamma}{T}\langle X_T, \Delta_{T} \rangle \right) \Big|\mathcal{F}_{T-1} \right] \cdot \exp\left(\frac{\gamma}{T} \sum_{t=1}^{T-1} \langle X_t, \Delta_t \rangle \right) \right]  \\
    & \leq \E\left[ \exp\left(\frac{\gamma^2 L^2}{T^2 N} + \frac{\gamma}{T} \sum_{t=1}^{T-1} \langle X_t, \Delta_t \rangle \right)  \right] \quad \textrm{(by \eqref{eq:aux1_Lemma_5})}\\
    & \leq ... \leq \exp\left(\frac{\gamma^2 L^2 }{TN}\right).
    \end{aligned}
    \]
It follows by Markov's inequality that 
    \[
    \Pr\left[\frac{1}{T}\sum_{t=1}^T \langle X_t, \Delta_t \rangle  \geq \delta\right] \leq \frac{\E[\exp(\frac{\gamma}{T} \sum_{t=1}^T \langle X_t, \Delta_t \rangle) ]}{\exp(\gamma \delta)} \leq 
    \exp\left(\frac{\gamma^2 L^2 }{TN} - \gamma \delta\right), \forall \gamma > 0.
    \]
    For a given confidence level $\Omega > 0$, applying the above relation with
  $\gamma = \frac{TN\delta}{2L^2}$ abd $\delta = \frac{2L \Omega}{\sqrt{TN}}$ yields the desired conclusion.
\end{proof}

\subsection{Duality Results} 
In this section we formalize the connection between the optimization problems \eqref{eq:main_problem} and \eqref{eq:dual}. Throughout this section we assume that $(\mathcal{Y},\mathcal{F},\nu)$ is a separable measure space in the sense of Exercise 10, Chapter 1 in \cite{stein2011functional}.

\label{subsection:appendix_duality}
\begin{prop}\label{prop:duality}
i) There exists a test $\varphi^*$ that solves \eqref{eq:main_problem}; that is, $\varphi^*$ maximizes  $\int \varphi g d\nu$
    among all level-$\alpha$ tests.  ii) Furthermore, there exists an optimizer $\kappa^*$ to the dual problem \eqref{eq:dual}, and the value of the dual problem is finite. iii) Moreover, for any solution $\hat{\kappa}$ of the dual \eqref{eq:dual}, and for any test $\hat\varphi$ that solves \eqref{eq:main_problem}, the pair $(\hat\varphi,\hat{\kappa})$ satisfy:
    
     \begin{equation}
    \label{eq:dual_cond_2}
    \begin{aligned}
        \hat{\varphi}(y) = 1, & \text{ when } g(y) > \sum_{m=1}^M \hat{\kappa}_m f_{m}(y), \\
        \hat{\varphi}(y) = 0, & \text{ when } g(y) < \sum_{m=1}^M \hat{\kappa}_m f_{m}(y),
    \end{aligned}
    \end{equation}
and the complementary slackness,
    \begin{equation}
    \label{eq:dual_cond_1}
    \hat{\kappa}_m \left( \int \hat{\varphi} f_{m}d\nu-  \alpha \right)=0, \forall m \in [M].
    \end{equation}

\end{prop}
\begin{proof}
Let $\nu$ denote the $\sigma$-finite measure defined over the measurable space $(\mathcal{Y},\mathcal{F})$. In a slight abuse of notation, denote by $L^\infty(\mathcal{Y})$ the set of essentially bounded real-valued measurable functions on $(\mathcal{Y},\mathcal{F})$. Let $L^{1}(\mathcal{Y})$ be the space of all real-valued measurable functions $f:\mathcal{Y} \rightarrow \mathbb{R}$ that are integrable with respect to $\nu$; that is $\int |f| d\nu < \infty$. Endow $L^\infty(\mathcal{Y})$ with the weak$^*$-topology; see \cite{RudinFA91}, p. 67, 68. By definition, a sequence $\{\varphi_n\}_{n \in \mathbb{N}} \subseteq L^\infty(\mathcal{Y})$ converges to $\varphi$ in the weak$^*$ topology if and only if 
    \[ \int  f \varphi_n d \nu \rightarrow \int f \varphi d\nu, \textrm{ for any } f \in L^1(\mathcal{Y}), \]
see p. 62-68 of \cite{RudinFA91}. It is known that when endowed with the weak$^*$ topology, the set $L^\infty(\mathcal{Y})$ is a linear topological space. 

\noindent{\scshape Proof of Statement i).} Define the set of all tests 
    \[
\mathcal{C} := \{ \varphi \in L^\infty(\mathcal{Y})  \mid 0 \le \varphi(y) \le 1 \text{ for } \nu\text{-a.e. } \},
\]
and consider the subset of all $\alpha$-level tests 
\[ \mathcal{C}_{\alpha}:=\left\{ \varphi \in \mathcal{C} \: \mid \: \int \varphi f_{m} d\nu \leq \alpha \textrm{ for all } m=1,\ldots, M  \right \}.\] 
Note that $\mathcal{C}_{\alpha}$ is nonempty since $\varphi_{0} \in \mathcal{C}_\alpha$ for any $\alpha$. By Lemma \ref{lem:test_set_cpact}, the set $\mathcal{C}_\alpha$ is compact under the weak$^*$-topology. As the objective function in \eqref{eq:main_problem} is continuous in the weak$^*$-topology, we conclude that there exists a test $\varphi^*$ that solves \eqref{eq:main_problem}.

\noindent{\scshape Proof of Statement ii).} 
    Recall the Lagrangian 
    \begin{equation*}
        L(\varphi, \kappa) = \int \varphi gd\nu -  \sum_{m=1}^M \kappa_m \left[\int \varphi f_{m}d\nu - \alpha\right].
    \end{equation*}
We first show that  Sion's minimax theorem holds, i.e.,
    \begin{equation}
    \label{eq:sion_minmax}
        \sup_{\varphi \in \mathcal{C}} \min_{\kappa \in \R^M_{+}} L(\varphi, \kappa) =  \min_{\kappa \in \R^M_{+}} \sup_{\varphi \in \mathcal{C}} L(\varphi, \kappa) := v.
\end{equation}
First, by Lemma \ref{lem:test_set_cpact},  $\mathcal{C}$ is a convex, compact subset of $L^{\infty}(\Yc)$ when endowed with the weak$^{*}$ topology. It is clear that $\R^M_{+}$ is a convex subset of $\R^M$. Second, since we endowed $L^\infty(\mathcal{Y})$ with the weak$^*$ topology, then, by definition, for any fixed $\kappa$, the functional $L(\cdot, \kappa): \mathcal{C} \to \R$ is a continuous functional, that is also linear. Similarly,   
or any fixed $\varphi$, $L(\varphi, \cdot):\R^M_{+} \to \R $ is a continuous and linear function of $\kappa$. Therefore, all conditions of Sion's minimax theorem are verified; see \citet[Theorem 3]{Simons1995}.

Since
\[ \sup_{\varphi \in \mathcal{C}_{\alpha}} \int \varphi g d\nu = \sup_{\varphi \in \mathcal{C}} \min_{\kappa \in \R^M_{+}} L(\varphi, \kappa), \]
and by part i) of Proposition \ref{prop:duality} there exists a test $\varphi^* \in C_{\alpha}$ such that
\[\int \varphi^* g d\nu = \sup_{\varphi \in \mathcal{C}_{\alpha}} \int \varphi g d\nu,\]
then   
    \[ 0  \leq v \leq 1.  \]
Therefore, by definition of minimum, there exists multipliers $\kappa^* \in \mathbb{R}^{M}_{+}$ such that 
\[0 \leq \sup_{\varphi \in \mathcal{C}} L(\varphi, \kappa^*) = \min_{\kappa \in \R^M_{+}} \sup_{\varphi \in \mathcal{C}} L(\varphi, \kappa) := v \leq 1.\]

\noindent{\scshape Proof of Statement iii).} Let $\hat{\kappa}$ be an arbitrary solution to the dual problem in \eqref{eq:dual}. Let  $\hat{\varphi}$ be an arbitrary solution to the  primal \eqref{eq:main_problem}. First, we would like to show that 
\begin{equation}\label{eq:prop1_statement_3}
\inf_{\kappa \in \R^M_{+}} L(\hat{\varphi}, \kappa) = \sup_{\varphi \in \mathcal{C}} \min_{\kappa \in \R^M_{+}} L(\varphi, \kappa), 
\end{equation}
which means that $\hat{\varphi}$ solves the maxmin problem. To this end, note that 
\[ \int \hat{\varphi} g d \nu = \sup_{\varphi \in \mathcal{C}_{\alpha}} \int \varphi g d\nu = \sup_{\varphi \in \mathcal{C}} \min_{\kappa \in \R^M_{+}} L(\varphi, \kappa). \]
Moreover, for any $\varphi \in \mathcal{C}_{\alpha}$, 
\[\min_{\kappa \in \R^M_{+}} L( \varphi , \kappa) = \int \varphi g d \nu.   \]
We conclude that 
\[  \min_{\kappa \in \R^M_{+}} L(\hat{\varphi}, \kappa) = \int \hat{\varphi} g d\nu = \sup_{\varphi \in \mathcal{C}_{\alpha}} \int \varphi g d \nu= \sup_{\varphi \in \mathcal{C}} \min_{\kappa \in \mathbb{R}_{+}^M} L(\varphi,\kappa). \]
This establishes \eqref{eq:prop1_statement_3}. By \eqref{eq:prop1_statement_3} and \eqref{eq:sion_minmax}, we have
    \begin{align*}
        L(\hat{\varphi}, \hat{\kappa}) 
    \geq \min_{\kappa \in \R^M_{+}} L(\hat{\varphi}, \kappa)
    & = \sup_{\varphi \in \mathcal{C}} \min_{\kappa \in \R^M_{+}} L(\varphi, \kappa)  \\
    & =  \min_{\kappa \in \R^M_{+}} \sup_{\varphi \in \mathcal{C}} L(\varphi, \kappa) 
    = \sup_{\varphi \in \mathcal{C}} L(\varphi,  \hat{\kappa}) 
    \geq L(\hat{\varphi}, \hat{\kappa}).
    \end{align*}
Then, note that  $L(\hat{\varphi}, \hat{\kappa}) = \inf_{\kappa \in \R^M_{+}} L(\hat{\varphi}, \kappa)$, implying \eqref{eq:dual_cond_1}. Also, $L(\hat{\varphi}, \hat{\kappa}) = \sup_{\varphi \in \mathcal{C}} L(\varphi,  \hat{\kappa}) $, implying \eqref{eq:dual_cond_2}.
\end{proof}
Below is a lemma proving that the domain of the primal problem is compact.
\begin{lemma}
\label{lem:test_set_cpact}
Let $(\mathcal{Y},\mathcal{F},\nu)$ be a separable measure space in the sense of Exercise 10, Chapter 1 in \cite{stein2011functional} and let $\nu$ be a $\sigma$-finite measure. Define
\[
\mathcal{C} := \{ \varphi \in L^{\infty}(\Yc) \mid 0 \le \varphi(y) \le 1 \text{ for } \nu\text{-a.e. } y \in \Yc\}, 
\]
and 
\[ \mathcal{C}_{\alpha}:=\left\{ \varphi \in \mathcal{C} \: \mid \: \int \varphi f_{m} d\nu \leq \alpha \textrm{ for all } m=1,\ldots, M  \right \}.\] 
Then, $\mathcal{C}$ and $\mathcal{C}_{\alpha}$ are compact in the weak$^*$ topology (where $L^{\infty}(\Yc)$ is viewed as the dual space of $L^{1}(\Yc)$). Moreover, $\mathcal{C}$ is a convex %weak$^*$ closed, and bounded 
subset of $L^{\infty}(\Yc)$.
\end{lemma}

\begin{proof}
Recall that $L^{\infty}(\Yc)$ is  identified with the dual of $L^{1}(\Yc)$, and the weak$^*$ topology on $L^{\infty}(\Yc)$ is the weakest topology that makes all maps 
\[
\varphi \mapsto \langle f, \varphi \rangle := \int_{\Yc} \varphi f \, d\nu,
\]
continuous for every $f\in L^{1}(\Yc)$. By the Banach--Alaoglu theorem \cite[p.68]{RudinFA91}, the closed unit ball
\[
\mathcal{B} := \left\{\varphi \in L^{\infty}(\Yc) \,\Big|\, \left|\int_{\Yc} \varphi f \, d\nu\right| \le 1 \text{ for all } f\in L^1(\Yc) \text{ with } \|f\|_1\le1\right\}
\]
is compact in the weak$^*$-topology. 

Observe that
\[
\mathcal{C} := \{ \varphi \in L^{\infty}(\Yc) \mid 0 \le \varphi(y) \le 1 \text{ for } \nu\text{-a.e. } y\in\Yc \}
\]
is a subset of $\mathcal{B}$, since for any $\varphi\in\mathcal{C}$ and any $f\in L^1(\Yc)$ with $\|f\|_1\le1$, one has
\[
\left|\int_{\Yc} \varphi f \, d\nu\right| \le \int_{\Yc} |\varphi| |f| \, d\nu \le \int_{\Yc} |f|\,d\nu \le 1.
\]

We now show that $\mathcal{C}$ is weak$^*$-sequentially closed. Suppose that $\{\varphi_n\}_{n=1}^{\infty}$ is a sequence in $\mathcal{C}$ that converges to some $\varphi \in L^\infty(\Yc)$ in the weak$^*$ topology. Assume for contradiction that $\varphi\notin\mathcal{C}$; then either the set $A_+ := \{ y\in\Yc \mid \varphi(y) > 1\}$
or $A_- := \{ y\in\Yc \mid \varphi(y) < 0\}$
has positive measure w.r.t. $\nu$. Without loss of generality, assume that $\nu(A_+)>0$. Define the function
\[
f(y) := \frac{\mathbf{1}_{A_+}(y)}{\nu(A_+)}.
\]
Then $f\in L^1(\Yc)$ and $\|f\|_1=1$. Since each $\varphi_n\in\mathcal{C}$, we have
\[
\int_{\Yc} \varphi_n f \, d\nu \le 1 \quad \text{for all } n.
\]
By the weak$^*$ convergence we obtain
\[
\lim_{n\to\infty} \int_{\Yc} \varphi_n f \, d\nu = \int_{\Yc} \varphi f \, d\nu \leq 1
\]
On the other hand, because $(\varphi - 1) f$ is positive we have 
\[
\int_{\Yc} (\varphi - 1) f \, d\nu \geq 0 \Longrightarrow \int_{\Yc} \varphi f \, d\nu \geq 1.
\]
This implies that $\int \varphi f d\nu = 1$, which in turn gives $\int(\varphi-1)f d\nu =0$. Such equality holds only when $(\varphi - 1) f = 0$ for $\nu$-almost surely. However, $(\varphi - 1) f > 0$ on $A_{+}$. This contradicts the fact that $\nu(A_+) > 0$!
This contradiction shows that $\varphi(y) \in [0,1]$ for $\nu$-almost every $y\in\Yc$, i.e., $\varphi\in\mathcal{C}$. Therefore, $\mathcal{C}$ is sequentially closed in the  weak$^*$-topology. Since $(\mathcal{Y},\mathcal{F},\nu)$ is a separable measure space, then $L^1(\mathcal{Y})$ is separable; see Exercise 10, Chapter 1 in \cite{stein2011functional}. Therefore, Theorem 3.16 in \cite{RudinFA91} p. 70 implies that $\mathcal{B}$ (with its subspace weak$^*$ topology) is compact and metrizable. This means that the sequential closure of $\mathcal{C}$  coincides with its closure; thus showing that $\mathcal{C}$ is closed in the weak$^*$ topology. Since $\mathcal{C}$ is a closed subset of the compact set $\mathcal{B}$, it is compact in the weak$^*$ topology. The proof that $C_{\alpha}$ is compact is entirely analogous and we omit it for the sake of brevity.

Finally, $\mathcal{C}$ is convex because if $\varphi_1, \varphi_2 \in \mathcal{C}$ and $t\in [0,1]$, then for $\nu$-almost every $y\in\Yc$,
\[
(1-t)\varphi_1(y)+t\varphi_2(y) \in [0,1],
\]
which implies that $(1-t)\varphi_1+t\varphi_2 \in \mathcal{C}$.
\end{proof}

\subsection{Theoretical Results on  $\varphi_{\bar{\kappa}_{T}}$}

\subsubsection{Asymptotic analyses of $\varphi_{\overline{\kappa}_{T}}$} \label{subsubsection:test_average_kappa}
\begin{lemma}\label{lem:asymptotic.convergence}
Suppose the conditions of Theorem \ref{thm:epsilon_least_favorable} hold. Then, we have $f(\overline{\kappa}_{T})\overset{p}{\rightarrow}\bar{\upsilon}$
as $T\rightarrow\infty$.    
\end{lemma}

\begin{proof}
For each $T$, let
\[
\varepsilon_{T}=\frac{2(1-\alpha)}{\alpha}\sqrt{\frac{\ln M}{T}},\eta_{T}=\frac{\alpha}{2(1-\alpha)^{2}}\varepsilon_{T}.
\]
Denote by $\left\{ \kappa_{T,j}\right\} _{j=1}^{T}$ the sequence
generated by Algorithm 1 with step number $T$ and step size $\eta_{T}$.
Then, $\overline{\kappa}_{T}=\frac{1}{T}\sum_{j=1}^{T}\kappa_{T,j}$.
It follows by Theorem 2 that, for any $\Omega>0$,
\[
\Pr\left\{ \left|f(\bar{\kappa}_{T})-\bar{\upsilon}\right|>\left(1+\frac{2\Omega}{\sqrt{(1-\alpha)^{2}N\ln M}}\right)\varepsilon_{T}\right\} <\exp\left(-\Omega^{2}\right).
\]
For each $\varepsilon>0$, pick $\Omega=\frac{\varepsilon\alpha}{8}\sqrt{NT}$.
Then, for all $T>\frac{\ln M}{\left(\frac{\varepsilon\alpha}{4(1-\alpha)}\right)^{2}}$,
we have
\[
\Pr\left\{ \left|f(\bar{\kappa}_{T})-\bar{\upsilon}\right|>\varepsilon\right\} <\exp\left(-\frac{\varepsilon^{2}\alpha^{2}NT}{64}\right),
\]
implying $f(\overline{\kappa}_{T})\overset{p}{\rightarrow}\bar{\upsilon}$
as $T\rightarrow\infty$.
\end{proof}
Lemma \ref{lem:asymptotic.convergence} shows that, $\overline{\kappa}_{T}$ is also asymptotically
a least favorable distribution. This result is expected given Theorem
\ref{thm:epsilon_least_favorable}'s finite-sample numerical convergence result. Next, we show that,
with additional regularity conditions, the Neyman-Pearson test $\varphi_{\overline{\kappa}_{T}}$
based on $\overline{\kappa}_{T}$ is asymptotically optimal as $T\rightarrow\infty$.

\begin{prop}
Suppose the conditions of Theorem \ref{thm:epsilon_least_favorable} hold. In addition, suppose for  all $\kappa\in\mathbb{R}_{M}^{+}$,
we have $g(y)-\sum_{m=1}^{M}\kappa_{m}f_{m}(y)\neq0$ for $\nu$-almost
all $y$. 
Then, the following statements are true:
\begin{enumerate}
    \item $\int\varphi_{\overline{\kappa}_{T}}gd\nu\overset{p}{\rightarrow}\bar{v}$
as $T\rightarrow\infty$;

\item For each convergent subsequence $\left\{ \overline{\kappa}_{T_{t}}\right\} $
of $\overline{\kappa}_{T}$, we have, for each $j\in[M]$, 
\[
\Pr\left\{\int\varphi_{\overline{\kappa}_{T_{t}}}(y)f_{j}d\nu\leq\alpha\right\}\rightarrow1,
\text{as } t\rightarrow\infty;
\]
\item If $\kappa^*$ is unique, we have, for each  $j\in[M]$,
\[
\int\varphi_{\overline{\kappa}_{T}}(y)f_{j}d\nu\overset{p}{\rightarrow}\int\varphi_{\kappa^{*}}f_{j}d\nu,\text{as } T\rightarrow\infty.
\]
\end{enumerate}
\end{prop}

\begin{proof}
First note, under the stated assumptions in the proposition, the absolute
continuity of $F_{j},j=1,\ldots, M$ and $G$ with respect to $\nu$
implies that, $g(y)\neq\sum_{m=1}^{M}\kappa_{m}^{*}f_{m}(y)$ for
$F_{j}$-almost all $y$, for each $j=1\ldots M$, and the same holds
for $G$-almost all $y$ as well. Together with Proposition \ref{prop:duality},
we have that, for any solution of the dual $\kappa^{*}$, the test
of form $\varphi_{\kappa^{*}}$, i.e., 
\begin{align*}
\varphi_{\kappa^{*}}(y)=1, & \text{ when }g(y)>\sum_{m=1}^{M}\kappa_{m}^{*}f_{m}(y),\\
\varphi_{\kappa^{*}}(y)=0, & \text{ when }g(y)\leq\sum_{m=1}^{M}\kappa_{m}^{*}f_{m}(y)
\end{align*}
is such that

\begin{equation}
\int\varphi_{\kappa^{*}}f_{m}d\nu\leq\alpha,\forall m\in[M],\quad\int\varphi_{\kappa^{*}}gd\nu=\bar{\upsilon}.\label{pf:asymptotic.opt.1}
\end{equation}
Moreover, dominated convergence theorem implies that the size function
\begin{align*}
\alpha_{j}(\cdotp) & =\int\varphi_{(\cdotp)}f_{j}d\nu:\mathbb{R}_{+}^{M}\rightarrow\mathbb{R}^{+},
\end{align*}
is continuous at all $\kappa\in\mathbb{R}_{M}^{+}$, for each $j=1,\ldots,M$,
and the power function 
\[
\pi(\cdotp)=\int\varphi_{(\cdotp)}gd\nu:\mathbb{R}^{M}\rightarrow\mathbb{R}^{+}
\]
is continuous at all $\kappa\in\mathbb{R}_{M}^{+}$ as well.

\noindent{\scshape Proof of Part 1 of the Theorem:}
As $\overline{\kappa}_{T}$ is bounded, Prohorov's Theorem (e.g.,
Theorem 2.4(ii) in \citealt{van2000asymptotic}) implies that there exists
a converging subsequence $\left\{ \overline{\kappa}_{T_{t}}\right\} $
such that $\overline{\kappa}_{T_{t}}\overset{d}{\rightarrow}X_{M}$
as $t\rightarrow\infty$, where $X_{M}$ is a random vector in $\mathbb{R}_{+}^{M}$.
Denote by $\mathcal{P}_{X_{M}}$ the probability measure for the distribution
of $X_{M}$. Since $f$ is continuous, continuous mapping theorem
implies that, as $t\rightarrow\infty$,$f(\overline{\kappa}_{T_{t}})\overset{d}{\rightarrow}f(X_{M})$. As we also know $f(\overline{\kappa}_{T})\overset{p}{\rightarrow}\bar{\upsilon}$
as $T\rightarrow\infty$, conclude that $f(\overline{\kappa}_{T})\overset{d}{\rightarrow}\bar{\upsilon}$,
implying that
$f(\overline{\kappa}_{T_{t}})\overset{d}{\rightarrow}\bar{\upsilon}$
as $t\rightarrow\infty$ as well. Therefore, $f(X_{M})$ must share
the same distribution as $\bar{\upsilon}$. Conclude that
$f(x_{M})=\bar{\upsilon}$,
for $\mathcal{P}_{X_{M}}$-almost every $x_{M}$. Since $\bar{\upsilon}$
is the optimal value, this implies that for $\mathcal{P}_{X_{M}}$-almost
every $x_{M}$, we have $f(x_{M})=\bar{\upsilon}=\inf_{\kappa\in\mathbb{R}_{+}^{M}}f(x)$,
i.e., $x_{M}$ solves the dual problem. Therefore, due to (\ref{pf:asymptotic.opt.1}),
we have $\int\varphi_{x_{M}}gd\nu=\bar{v}$ for $\mathcal{P}_{X_{M}}$-almost
every $x_{M}$. Conclude that $\int\varphi_{X_{M}}gd\nu=\bar{\upsilon}$
with probability 1. By continuity of the power function $\pi(\cdotp)=\int\varphi_{(\cdotp)}gd\nu$ in  $\mathbb{R}_{+}^{M}$, 
conclude further that as $t\rightarrow\infty$, 
\[
\pi(\overline{\kappa}_{T_{t}})\overset{d}{\rightarrow}\int\varphi_{X_{M}}gd\nu=\bar{v}.
\]
As the preceding convergence claim holds for every convergent subsequence,
conclude that $\pi(\overline{\kappa}_{T})\overset{d}{\rightarrow}\bar{v},$ as
$T\rightarrow\infty$, implying $\pi(\overline{\kappa}_{T})\overset{p}{\rightarrow}\bar{v}$. 

\noindent\noindent{\scshape Proof of Part 2 of the Theorem:}
Analogous to the proof of part 1 of the theorem, consider a convergent subsequence
$\left\{ \overline{\kappa}_{T_{t}}\right\} $ that converges in distribution
to some random vector $X_{M}\in\mathbb{R}_{+}^{M}$ with a probability
measure $\mathcal{P}_{X_{M}}$ for its distribution function. Note,
by analogous arguments to the proof of part 1, for $\mathcal{P}_{X_{M}}$-almost
every $x_{M}$, we have
\begin{align}
\int\varphi_{x_{M}}f_{j}d\nu & \leq\alpha,\forall j\in[M].\label{eq:1-1}
\end{align}
Therefore, $\int\varphi_{X_{M}}f_{j}d\nu\leq\alpha$ with probability
1 for each $j\in[M]$. The proof is further divided in three steps.

\noindent {\scshape Step 1}: We show that 
$\alpha_{j}(\overline{\kappa}_{T_{t}})\overset{p}{\rightarrow}\int\varphi_{X_{M}}f_{j}d\nu$
for each $j\in[M]$. Note since $\alpha_{j}$ is bounded and continuous,
Portmanteau's Lemma implies that 
\[
\mathbb{E}\alpha_{j}(\overline{\kappa}_{T_{t}})\overset{}{\rightarrow}\mathbb{E}\int\varphi_{X_{M}}f_{j}d\nu\leq\alpha
\]
as $t\rightarrow\infty$. As $\alpha_{j}(\overline{\kappa}_{T_{t}})$
is bounded, $\alpha_{j}(\overline{\kappa}_{T_{t}})$ is also uniformly
integrable. Therefore, we have
\[
\alpha_{j}(\overline{\kappa}_{T_{t}})\overset{p}{\rightarrow}\int\varphi_{X_{M}}f_{j}d\nu
\]
as $t\rightarrow\infty$ for each $j\in[M]$. 

\noindent {\scshape Step 2}: We show that for any $\epsilon>0$, we have, as $t\rightarrow\infty$,
$\Pr\{\alpha_{j}(\overline{\kappa}_{T_{t}})\leq\alpha+\epsilon\}\rightarrow1$. For any $\epsilon>0$, it suffices to show that $\Pr\{\alpha_{j}(\overline{\kappa}_{T_{t}})>\alpha+\epsilon\}\rightarrow0$
as $t\rightarrow\infty$. To this end, note for any $0<\delta<\epsilon$:
\begin{align*}
 & \Pr\{\alpha_{j}(\overline{\kappa}_{T_{t}})>\alpha+\epsilon\}\\
= & \Pr\{\alpha_{j}(\overline{\kappa}_{T_{t}})>\alpha+\epsilon,\alpha_{j}(\overline{\kappa}_{T_{t}})-\int\varphi_{X_{M}}f_{j}d\nu>\delta\}\\
+ & \Pr\{\alpha_{j}(\overline{\kappa}_{T_{t}})>\alpha+\epsilon,\alpha_{j}(\overline{\kappa}_{T_{t}})-\int\varphi_{X_{M}}f_{j}d\nu\leq\delta\}\\
\leq & \Pr\left\{\alpha_{j}(\overline{\kappa}_{T_{t}})-\int\varphi_{X_{M}}f_{j}d\nu>\delta\right\}\\
+ & \Pr\left\{\int\varphi_{X_{M}}f_{j}d\nu>\alpha+\epsilon-\delta\right\}.
\end{align*}
Note $\Pr\{\alpha_{j}(\overline{\kappa}_{T_{t}})-\int\varphi_{X_{M}}f_{j}d\nu>\delta\}\rightarrow0$
as $t\rightarrow\infty$ by the conclusion from step 1, and $\Pr\{\int\varphi_{X_{M}}f_{j}d\nu>\alpha+\epsilon-\delta\}=0$
since $\epsilon-\delta>0.$ Conclude that
$\Pr\{\alpha_{j}(\overline{\kappa}_{T_{t}})>\alpha+\epsilon\}\rightarrow0$
as $t\rightarrow\infty$.

\noindent {\scshape Step 3}: From step 2, we have that, for each $\epsilon>0$, $\Pr\{\alpha_{j}(\overline{\kappa}_{T_{t}})\leq\alpha+\epsilon\}\rightarrow1$
as $t\rightarrow\infty$. As $\epsilon$ is arbitrary, conclude that 
$\Pr\{\alpha_{j}(\overline{\kappa}_{T_{t}})\leq\alpha\}\rightarrow1$,
as $t\rightarrow\infty$ as desired.

\noindent{\scshape Proof of Part 3 of the Theorem:}
Since $f(\overline{\kappa}_{T})\overset{p}{\rightarrow}\bar{v}=f(\kappa^{*})$
as $T\rightarrow\infty$ and $\kappa^{*}$ is the unique solution of the dual, we must have $\overline{\kappa}_{T}\overset{p}{\rightarrow}\kappa^{*}$
as well given continuity of $f$. The conclusion then follows immediately from continuous mapping
theorem. 
\end{proof}

\end{appendix}

\end{document}